\documentclass[a4paper,USenglish,cleveref, autoref, thm-restate]{lipics-v2021}
\usepackage{booktabs}
\usepackage{xspace}
\usepackage{thmtools}
\usepackage{thm-restate}
\usepackage{microtype}
\usepackage{gensymb}
\usepackage{enumitem}
\usepackage{optidef}
\usepackage[noabbrev,capitalise]{cleveref}

\graphicspath{{graphics}}

\newcommand\polylog{{\rm polylog}}
\newcommand\argmin{{\rm argmin}}
\newcommand\argmax{{\rm argmax}}
\newcommand\opt{{\rm opt}}
\newcommand\sopt{{s_{\opt}}}
\newcommand\Max{{M_{\max}}}
\newcommand\Mis{{M_{\rm mis}}}
\newcommand{\pluseq}{\mathrel{+}=}
\newcommand{\dist}{\ensuremath{\mathit{dist}}\xspace}

\newcommand{\Matousek}{Matou{\v s}ek\xspace}

\newcommand{\mathset}[1]{\ensuremath {\mathbb {#1}}\xspace}
\newcommand{\eps}{\ensuremath{\varepsilon}\xspace}
\newcommand{\R}{\mathset {R}}

\newcommand{\mkmcal}[1]{\ensuremath{\mathcal{#1}}\xspace}
\newcommand{\A}{\mkmcal{A}}
\newcommand{\etal}{\textnormal{et al.}\xspace}

\newcommand{\CH}{\ensuremath{\it CH}\xspace}

\newcommand{\enumit}[1]{\textcolor{darkgray}{\sffamily\bfseries\upshape\mathversion{bold}#1}}

\title{Robust Classification of Dynamic Bichromatic Point Sets in $\R^2$}

\author{Erwin Glazenburg}{Department of Information and Computing
  Sciences, Utrecht University, The
  Netherlands}{e.p.glazenburg@uu.nl}{}{Supported by the Netherlands Organisation for Scientific Research
(NWO) under project OCENW.M20.135.
}
\author{Marc van Kreveld}{Department of Information and Computing Sciences, Utrecht University, The Netherlands}{m.j.vankreveld@uu.nl}{}{}
\author{Frank Staals}{Department of Information and Computing Sciences, Utrecht University, The Netherlands}{f.staals@uu.nl}{}{}

\authorrunning{E. Glazenburg, M. van Kreveld, F. Staals}
\Copyright{Erwin Glazenburg, Marc van Kreveld, and Frank Staals}



\nolinenumbers 

\begin{document}

\maketitle

\begin{abstract}
  Let $R \cup B$ be a set of $n$ points in $\R^2$, and let
  $k \in 1..n$. Our goal is to compute a line that ``best'' separates
  the ``red'' points $R$ from the ``blue'' points $B$ with at most $k$
  outliers. We present an efficient semi-online dynamic data structure
  that can maintain whether such a separator exists. Furthermore, we
  present efficient exact and approximation algorithms that compute a
  linear separator that is guaranteed to misclassify at most $k$,
  points and minimizes the distance to the farthest outlier. Our exact
  algorithm runs in $O(nk + n \log n)$ time, and our
  $(1+\eps)$-approximation algorithm runs in
  $O(\eps^{-1/2}((n + k^2) \log n))$ time. Based on our
  $(1+\eps)$-approximation algorithm we then also obtain a semi-online
  data structure to maintain such a separator efficiently.
\end{abstract}

\clearpage

\section{Introduction}

Classification is a well known and well studied problem: given a
``training'' set of $n$ data items with known classes, decide which
class to assign to a new query item. Support Vector Machines
(SVMs)~\cite{svm} are a popular method for binary classification in
which there are just two classes: \emph{red} and \emph{blue}. An SVM
maps the input data items to points in $\R^d$, and constructs a
hyperplane $s$ that separates the red points $R$ from the blue points
$B$ ``as well as possible''. Intuitively, it tries to minimize the
distance from $s$ to the set $X(s,B \cup R) \subseteq R\cup B$ of
points \emph{misclassified} by $s$ while maximizing the distance to
the closest correctly classified points. A red point $r \in R$ is
misclassified if it lies strictly inside the halfspace $s^+$ above
(left) of $s$, whereas a blue point $b \in B$ is misclassified if it
lies strictly inside the halfspace $s^-$ below $s$. See
Figure~\ref{fig:higherKLowerError} for an illustration. An SVM is typically
modeled as a convex quadratic programming problem with linear
constraints. However, this cannot provide guarantees on the number of
misclassifications nor on the running time\footnote{When we restrict
  the coefficients in the SVM formulation to be rational numbers with
  bounded bit complexity such a problem can be solved in polynomial
  time~\cite{goldfarb1990qp,monteiro1989qp,kozlov1980}. However, it is
  unclear if they can be extended to allow for arbitrary real valued
  costs.}. In practice, solving such optimization problems is
possible, but it is computationally expensive as it involves $n+d$
variables~\cite{hsieh14}. This problem is magnified as training a
high-quality classifier typically requires computing many classifiers,
each trained on a large subset of the input data, during
cross-validation. Similarly, in streaming settings, the labeled input
points arrive on the fly, and old data points should be removed due to
concept drift~\cite{schlimmer86beyond}. Each such update requires
recomputing the classifier. Hence, this limits the applicability of
SVMs in these settings, even when the input data is just
low-dimensional.

\begin{figure}
    \centering
    \includegraphics{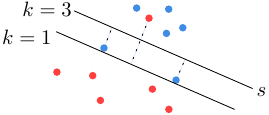}
    \caption{Red and blue points, and two optimal separators for
      $\Max$. Allowing more misclassifications could decrease the
      value of $\Max$.}
    \label{fig:higherKLowerError}
\end{figure}

\subparagraph{The goal.} We aim to tackle both these problems. That
is, we wish to develop an ``SVM-like'' linear classifier that can
provide guarantees on the number of misclassified points $k$, and can
be constructed and updated efficiently. As the problem of minimizing
$k$ is NP-complete in
general~\cite{amaldi95compl_approx_findin_maxim_feasib}, we restrict
our attention to the setting where the input points are
low-dimensional. As it turns out, even for points in the plane, this
is a challenging problem.

For a separator $s$, let $\Mis(s) = |X(s, B \cup R)|$ be the number of
points misclassified by $s$. Let
$S_k(B \cup R) = \{ s \mid \Mis(s) \leq k \}$ denote the set of
hyperplanes that misclassify at most $k$ points from $B \cup R$, and
let $\dist(p,q)$ denote the Euclidean distance between geometric
objects $p$ and $q$. When the point sets are linearly separable, we
want to compute a maximum-margin separator
$s_{\textit{strip}} \in S_0(R\cup B)$ that correctly classifies all
points and maximizes the distance
$M_{\mathit{strip}}(s_{\textit{strip}}) = \min_{p \in R\cup B}
\dist(s_{\textit{strip}},p)$ to the closest points, exactly as in an
SVM. Moreover, we would like to efficiently maintain such a separator
when we insert or delete a point from $B \cup R$. We can relatively
easily find such a separator by maintaining the convex hulls of $R$
and $B$ (Theorem~\ref{thm:2d_maxmargin}). So the main challenge occurs
when the point sets are not linearly separable. In this case, given a
maximum number of allowed misclassified points $k$, our aim is to find
a separator $s_\opt \in S_k(R\cup B)$ that minimizes the (Euclidean)
distance $\Max(s) = \max_{p \in X(s,B \cup R)} \dist(p,s)$ to the
furthest misclassified point. This thus asks for a minimum width strip
containing the $k$ outliers. We again would like to maintain such a
separator when points are inserted or deleted. Furthermore, we may
want to compute the smallest number $k_{\min}$ for which there exists
a separator $s_{\min}$ that misclassifies at most $k_{\min}$
points. Note that decreasing the number of outliers may increase the
value of $\Max$, i.e. when $k_{\min} < k$ we may have
$\Max(s_{\min}) > \Max(s_\opt)$, see
Figure~\ref{fig:higherKLowerError}.

By the above discussion we distinguish four general variations of the
problem:
\begin{description}
    \item[MaxStrip:] find a separator $s_{\textit{strip}} = \argmax_{s \in S_0(R\cup B)} M_{\mathit{strip}}(s)$
    \item[MinMax:] find a separator $s_{\max} = \argmin_s \Max(s)$
    \item[MinMis:] find a separator $s_{\textrm{mis}} = \argmin_s \Mis(s)$
    \item[$k$-mis MinMax:] given a value $k$, find a separator $s_{\opt} = \argmin_{s \in S_k(B \cup R)} \Max(s)$
\end{description}

\subparagraph{Related Work.} It is well known that for points in
$\R^d$, for constant $d$, we can test if $R$ and $B$ can be linearly
separated in $O(n)$ time by using linear programming
(LP)~\cite{megiddo84linear_progr_linear_time_when}. The problem
becomes much more challenging when we wish to allow for a limited
number of misclassifications. Everett, Robert, and van
Kreveld~\cite{lessThanK} showed that for point sets $R$ and $B$ in the
plane, one can find a line that separates $R$ and $B$ while allowing
for at most $k$ misclassifications in $O(n\log n + nk)$
time. \Matousek~\cite{matousek95geomet_optim_few_violat_const} showed
how to solve LP-type problems while allowing at most $k$ violated
constraints. In particular, for linear programming in $\R^2$, and thus
also for our problem of finding a separating line with at most $k$
misclassifications, his algorithm runs in $O(n\log n + k^3\log^2 n)$
time. Chan~\cite{chanLPViolations} improves 
this to $O((n+k^2)\log n)$ time, and can compute the smallest number
$k$ for which the points can be separated (the MinMis
problem) in the same time. Aronov
\etal~\cite{aronov12minim} considered computing optimal separators
with respect to other error measures as well. In particular, they
consider minimizing the distance $\Max(s)$ from $s$ to the furthest
misclassified point, as well as minimizing the average (squared)
distance to a misclassified point
$M^\beta_{\textrm{avg}}(s) = \sum_{p \in X(s,B \cup R)} (\dist(s,
p))^\beta$. For $n$ points in $\R^2$, their running times for
computing an optimal separator vary from $O(n\log n)$ for the $\Max$
measure (the MinMax problem), to $O(n^{4/3})$ for the
$M^1_{\textrm{avg}}$ measure, to $O(n^2)$ for $\Mis$ (the MinMis
problem) and the $M^2_{\textrm{avg}}$ measures. Some of their results
extend to points in higher dimensions as well. Har-Peled and
Koltun~\cite{harpeled} consider similar measures, and present both
exact and approximation algorithms. For example, they present an exact
$O(nk^{d+1}\log n)$ time algorithm to find a hyperplane that minimizes
the number of outliers (for points in $\R^d$), and an
$O(n(\eps^{-2}\log n)^{d+1})$ time algorithm to compute a
$(1+\eps)$-approximation of that number\footnote{Here and throughout
  the rest of the paper, $\eps > 0$ is an arbitrarily small
  constant.}. Their exact and approximation algorithms for computing a
hyperplane minimizing $\Max$ run in $O(n^d)$ and $O(n\eps^{(d-1)/2})$
time, respectively. Matheny and
Phillips~\cite{matheny21approx_maxim_halfs_discr} consider computing a
separating hyperplane $s$, so that the discrepancy (that is, the
fraction of red points in $s^-$ minus the fraction of blue points in
$s^-$) is maximized. They present an $O(n+\eps^{-d}\log^4(\eps^{-1}))$
time algorithm that makes an additive error of at most $\eps$
(and thus ``misclassifies'' at most $\eps n$ points more than an
optimal (with respect to discrepancy) classifier).

\subparagraph{Results.} As a warmup, in
Section~\ref{sec:fully_dynamic_1d}, we show that for points in $\R^1$
we can achieve both our goals: minimizing $\Max$ with a hard
guarantee on the number of outliers \emph{and} efficiently supporting
updates. In particular, we present an optimal linear space solution:

\begin{restatable}{theorem}{optimalOneD}
  \label{thm:1d_optimal}
  Let $B \cup R$ be a set of $n$ points in $\R^1$. There is an $O(n)$
  space data structure that, given a query value $k \in 1..n$ can
  compute an optimal separator $s_\opt \in S_k(R\cup B)$ with respect
  to $\Max$ in $O(\log n)$ time, and supports inserting or deleting a
  point in $O(\log n)$ time.
\end{restatable}

The main focus of our paper is to establish whether we can achieve
similar results for points in $\R^2$. If the points are separable, we
can maintain a maximum-margin
separator---essentially a maximum width strip---in $O(\log^2 n)$
time per update, see
Section~\ref{sec:preliminaries}. 

The problem gets significantly more complicated when the point sets
are not separable, and we thus wish to compute, and maintain, a
separator $s_\opt \in S_k(B\cup R)$ minimizing the distance $\Max$ to
the farthest misclassified point. We can test whether a
separator $s \in S_k(B \cup R)$ exists (and find
the smallest $k$ for which a separator exists) using LP with violations. In
Section~\ref{sec:Linear_Programming_with_Violations} we show
how to dynamize Chan's approach to maintain such a separator
when the set of points changes. In particular, given a static linear
objective function $f : \R^2 \to \R$ and a dynamic set $H$ of
halfplanes that is given in a semi-online manner, i.e. at the time we
insert a halfplane $h$ into $H$ we are told when we will delete $h$,
we show how to efficiently maintain a point $p$ minimizing $f$ that
lies outside at most $k$ halfplanes from $H$:

\begin{restatable}{theorem}{linearProgramming}
  \label{thm:linear_programming}
  Let $H$ be a set of $n$ halfspaces in $\R^2$, let $f$ be a linear
  objective function, and let $k \in 1..n$. There is an
  $O(n + k^2\log^2 n)$ space data structure that maintains a point
  $p$ that violates at most $k$ constraints of $H$ (if it exists)
  and minimizes $f$, and supports semi-online updates in expected amortized
  $O(k\log^3 n)$ time.
\end{restatable}

This then also allows us to maintain whether a separator that
misclassifies at most $k$ points exists in amortized $O(k\log^3 n)$
time per (semi-online) update, as well as maintain the minimum value
$k$ for which this is the case. Since linear programming queries have
many other applications, e.g.\ finding extremal points and tangents,
we believe this result to be of independent interest. For example,
given a threshold $\delta$, our data structure also allows us to
maintain a line that minimizes the number of points $k$ at vertical
distance exceeding $\delta$ from $\ell$ in amortized $O(k\log^3 n)$
time per update. Note that the best update time we can reasonably
expect with this approach is $O((1+k^2/n)\log n)$. For values of $k$
that are small (e.g.\ polylogarithmic) or very large (near linear) our
approach is relatively close to this bound.

In Section~\ref{sec:An_exact_algorithm_in_2d}, we then actually
incorporate finding the best separator from $S_k(B\cup R)$; i.e.\ a
separator that minimizes $M_{\max}$. We first tackle the algorithmic
problem of simply computing such an optimal separator. Our main result
here is:

\begin{restatable}{theorem}{algorithmTwoD}
  \label{thm:2d_algorithm}
  Let $B \cup R$ be a set of $n$ points in $\R^2$, and let
  $k \in 1..n$. We can compute a separator $s_\opt \in S_k(B \cup R)$
  minimizing $\Max$ in
  \begin{itemize}
  \item $O(nk + n\log n)$ time,
  \item $O((n+|S_k(B\cup R)|+n+k^3)\log^2 n)$ time, or
  \item when $k=k_{\min}$ in $O(k^{4/3} n^{2/3} \log n +
    (n + k^2) \log n)$ time.
  \end{itemize}
\end{restatable}

The key challenge here is that (the region in the dual plane
representing) $S_k(B \cup R)$, may consist of $\Theta(k^2)$ connected
components, and each one has very little structure. Where in the
linear programming approach we can efficiently find one local minimum
per connected component, that is no longer the case here. Instead,
explicitly construct the boundary of this region. Unfortunately, the
total complexity of $S_k(B \cup R)$ is rather large:
Chan~\cite{chan10bichromatic} gives an upper bound of
$O(nk^{1/3}+n^{5/6-\eps}k^{2/3+2\eps}+k^2)$. We give two different
algorithms to construct $S_k(B \cup R)$, and then efficiently find an
optimal separator. When we restrict to the case where $k=k_{\min}$,
i.e. finding an separator that minimizes $\Max$ among all lines that
misclassify the least possible number of outliers, each connected
component of $S_k(B\cup R)$ is a single face in an arrangement of
lines. This then gives us a slightly faster
$O(k^{4/3} n^{2/3} \log^{2/3} (n / k) + (n + k^2) \log n)$ time
algorithm as well.

Unfortunately, even when $k=k_{\min}$, dynamization turns out to be
extremely challenging. In Section~\ref{sec:insertion_only_tight_k} we
present an $O((k^{4/3} n^{2/3}+n) \log^5 n)$ space data structure that
supports insertions in amortized $O(kn^{3/4 + \eps})$
time, provided that the convex hulls of
$R$ and $B$ remain the same. While the applicability of this result is limited, we do use and develop an interesting combination of
techniques here. For example; we develop a near linear space data
structure that stores the lower envelope of surfaces, and allows
for sub-linear time vertical ray shooting queries.

In Section~\ref{sec:eps-Approximation}, we slightly relax our goal and
consider approximating the distance $\Max$ instead. Our key idea is to
replace the Euclidean distance by a convex distance function. This
avoids some algebraic issues, as the distance between a point and
a line now no longer has a quadratic dependency on the slope of the
line. Instead the dependency becomes linear. We now obtain a much more
efficient algorithm for finding a good separator:

\begin{restatable}{theorem}{approximationAlgorithm}
  \label{thm:2d_approximation_algorithm}
  Let $B \cup R$ be a set of $n$ points in $\R^2$, let $k \in 1..n$,
  and let $\eps > 0$.  We can compute a separator
  $s \in S_k(B \cup R)$ that is a $(1+\eps)$ approximation with
  respect to $\Max$ in $O(\eps^{-1/2}((n + k^2) \log n))$ time.
\end{restatable}

Our approach essentially ``guesses'' the width $\delta$ of a strip
``separating'' the point sets. We then show that we can use the linear
programming machinery to efficiently test whether there exists such a
strip containing at most $k$ outliers. This
involves extending the algorithm to deal with both ``soft
constraints'' (that may be violated) as well as ``hard constraints''
that cannot be violated. We then find the smallest $\delta$ for which
such a strip exists using parametric
search~\cite{megiddo1983parametric}. This leads to the following
result:

\begin{restatable}{theorem}{dynamicApproximation}
  \label{thm:2d_approximate_ds}
  Let $B \cup R$ be a set of $n$ points in $\R^2$, let $k \in 1..n$,
  and let $\eps > 0$. There is an
  $O(\eps^{-1/2}(k^2 \log^2 n + n))$ space data structure that maintains
  a separator $s \in S_k(B \cup R)$ that is a $(1+\eps)$-approximation with respect to $\Max$, and supports semi-online
  updates in expected amortized $O(\eps^{-1/2}k\log^3 n \log \log k)$ time.
\end{restatable}

\subparagraph{Applications.} Our data structure from
Theorem~\ref{thm:2d_approximate_ds} can reduce the total time in a
leave-out-one cross validation process by roughly a linear factor in
comparison to Theorem~\ref{thm:2d_approximation_algorithm}. For
$m$-fold cross validation we gain a factor $m$. Similarly, in a
streaming setting in which we maintain a window of width $w$, we gain
a factor of roughly $w$. Note that in both these settings, the
semi-online updates indeed suffice.



\subparagraph{Open Problems.} There are several remaining interesting
open problems. For example, can we improve the running time of the
algorithm from Theorem~\ref{thm:2d_algorithm}? In particular, can we
achieve $O((n+k^2)\log n)$ time (the time it takes to test if a valid
separator exists)? This will also make it easier to then turn the
algorithm into a dynamic data structure. Furthermore, it would be
interesting to develop a fully dynamic data structure for LP with
violations. And of course, can we obtain similar results in higher
dimensions?



\section{A fully dynamic solution for points in \texorpdfstring{$\R^1$}{R1}}
\label{sec:fully_dynamic_1d}

In this section we prove Theorem~\ref{thm:1d_optimal}. We present a
linear space data structure to maintain an optimal separator $s_\opt$
that that classifies points to its left as red and points to its right
as blue. If we maintain both this data structure and its mirrored
version, which classifies points to its left as blue and points to its
right as red, the best of the two will be the optimal separator.

We first describe how to compute and maintain a separator
$s_{\max} \in S_n(B \cup R)$ that minimizes $\Max$ (i.e. the MinMax
problem). We then describe how to solve the MinMis problem:
i.e. maintain a separator $s_{\min}$ that minimizes the number of
outliers. In both cases we obtain linear space solutions with $O(\log
n)$ update time. We then combine these results to maintain an optimal
separator $s_\opt \in S_k(B \cup R)$.

\subparagraph{Minimizing $\Max$.} We start with the simple case in
which we wish to compute a separator
$s_{\max} \in S_n(B \cup R)$ that minimizes $\Max$.

If $R$ and $B$ are separable, this means $R$ lies fully to the left of
$B$. The optimal separator has the largest distance to both
sets. Clearly the optimal separator then lies exactly between the
rightmost point of $R$ and the leftmost point of $B$. By maintaining
the extremal points of $R$ and $B$ in a heap, we can thus find an
optimal separator in $O(1)$ time: we simply query the heaps and take
the average.

If $R$ and $B$ are inseparable, the rightmost red point is to the
right of the leftmost blue point, but the optimal separator $s_{\max}$
that minimizes the maximum distance to any misclassified point is
still exactly in the middle of them. This means the above approach
still works without any modifications.

\begin{figure}
\centering
\includegraphics{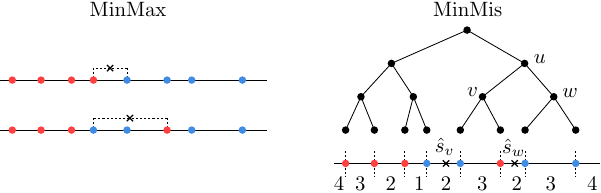}
\caption{Optimal separators for the MinMax problem (left), and the
  binary tree for the MinMis problem (right). Node $u$ has children
  $v$ and $w$ with a local optima $\hat{s}_v$ and $\hat{s}_w$.}
\label{fig:1D_intervals}
\end{figure}

\subparagraph{Minimizing $\Mis$.} Next, we consider computing a
separator $s_{\min}$ that minimizes $\Mis$. We again present an $O(n)$
space solution supporting updates in $O(\log n)$ time.

The red and blue points $P = R\cup B$ partition $\R$ into $n + 1$
intervals. All separators within a given interval misclassify the same
number of points, regardless of where in the interval they lie
exactly. So we can sort all points in $O(n \log n)$ time to obtain all
$O(n)$ intervals. A separator in the leftmost interval lies left of
all points in $P$. It misclassifies all red points and correctly
classifies all blue point, so it misclassifies $|R|$ points. When we
move $s$ to the right over a red point, the number of
misclassifications decreases by one, whereas moving over a blue point
increases the number of misclassifications. We can thus simply scan
through the points while maintaining the number of misclassifications.

We can maintain this minimum by maintaining a balanced binary search
tree of $P$ with some additional information.

Firstly in each node $u$ with subtree $T_u$ we store the number of
blue points $B_u$ and the number of red points $R_u$ in $T_u$.

Secondly in each node $u$ with subtree $T_u$ we consider the local
problem of finding a separator $s_u$ within this subtree (including
the intervals adjacent to the leftmost and rightmost points of $T_u$)
that minimizes the number of misclassified points in $T_u$. Note that
if $u$ is the root this problem is equal to MinMis. We store the
optimal value for this problem in each node as $M^u_{\rm mis}$.

With this tree we can find the minimum value of $\Mis(s)$ in $O(1)$
time, since this is stored in the root of the tree. We can also find a
minimum separator $\hat{s}$ in $O(\log n)$ time by going down the
tree, always choosing a child that contains a minimum separator. 

Dynamically maintaining $R_u$ and $B_u$ is easy, but to maintain
$M^u_{\rm mis}$ we first need to prove the following lemma.

\begin{lemma}
\label{lem:childOptimum}
Let $u$ be a node with left child $v$ and right child $w$, and let $\hat{s}_v$ and $\hat{s}_w$ be optimal local separators for $v$ and $w$ respectively. Then either $\hat{s}_v$ or $\hat{s}_w$ is also an optimal local separator for $u$.
\end{lemma}

\begin{proof}
Proof by contradiction: assume neither $\hat{s}_v$ nor $\hat{s}_w$ are optimal local optimal separators for $T_u$, but there is some optimal $s'$ which lies in $T_v$ (or symmetrically in $T_w$). Note that the number of extra misclassifications due to adding $T_w$ is constant for any local separator $s$ in $T_v$, since $T_w$ will always be fully on the right of $s$. So if $s'$ is a better separator than $\hat{s}_v$ for $T_u$, it is also a better separator than $\hat{s}_v$ for $T_v$. This is a contradiction since $\hat{s}_v$ should be an optimal separator for $v$.
\end{proof}

For a node $u$ with children $v$ and $w$, whenever $M^v_{\rm mis}$ or
$M^w_{\rm mis}$ changes we can update $M^u_{\rm mis}$ as follows. By
\cref{lem:childOptimum} we know that either $\hat{s}_v$ or $\hat{s}_w$
is an optimal separator for $T_u$. For $\hat{s}_v$ the total number of
misclassified points in $T_u$ is $M^v_{\rm mis} + R_w$, and similarly
for $\hat{s}_w$ it is $M^w_{\rm mis} + B_v$. We store the minimum of
the two as $M^u_{\rm mis}$.

This update can be done in constant time per node so the total update
time will remain $O(\log n)$, since we only need to update nodes along
the update path and nodes that are adjusted during rebalancing.

\subparagraph{Minimizing $\Max$ with at most $k$ misclassifications.}
Let $s_{max}$ be the optimal separator for the MinMax problem with
value $\Max(s_{max})$. Then any separator $s'$ has value
$\Max(s_{max}) + ||s' - s_{max}||$, since it is moved
$\|s' - s_{max}\|$ farther away from one of the two extremal points.

\begin{observation}
  The valid separator $s_\opt$ (i.e. $s_\opt \in S_k(B \cup R)$) with
  the smallest distance to $s_{max}$ is an optimal separator for
  $k$-mis MinMax.
\end{observation} 

So, if $s_{max}$ is valid then $s_{max}$ itself is the optimal
separator. Furthermore, recall that points exactly on the separator
are always classified correctly, both red and blue. This means the
optimal separator will either be $s_{max}$ or be on an input point
between a valid and an invalid interval; a point fully within a valid
interval can be moved towards $s_{max}$ to decrease its distance to
the extremal point.

Hence, this suggests that we can compute $s_\opt$ (if it exists) in
$O(n\log n)$ time as follows. We first find the optimal MinMax
separator $s_{\max}$ in $O(n)$ time, and then as for the static MinMis
problem we compute for every interval if it is valid or not. If
$s_{\max}$ is valid, we return it. Otherwise for each point
$p \in B \cup R$ that lies between a valid and an invalid interval we
compute $\dist(p, s_{\max})$, and maintain and return the point with
the smallest distance. If no such point exists, $s_\opt$ does not exist.

To turn the above algorithm into a data structure, we maintain the
same heaps as for the MinMax problem. We also maintain the same
augmented balanced binary tree as for the MinMis problem, but include
an additional symbolic (colorless) point $-\infty$, such that every
interval has a point on its left. See
Figure~\ref{fig:kMisMinMaxWalk}. Updating these structures takes
$O(\log n)$ time.

To find an optimal separator $s_\opt \in S_k(B\cup R)$, we first we
query the MinMis data structure to check if there exists any separator
that misclassifies at most $k$ points. If not, $s_\opt$ does not
exist. Otherwise, we query the heaps for the optimal MinMax separator
$s_{\max}$ in $O(1)$ time. We then search for the point directly left
of $s_{\max}$ in our binary tree. The search path is shown in green in
Figure \ref{fig:kMisMinMaxWalk}. For each node $u$ along the search
path we compute $k_u$, the number of points outside of $T_u$
misclassified by $s_{\max}$. Observe that if $u$ is the root, then
$k_u = 0$, and if $u$ is a leaf, then $k_u = \Mis(s_{\max})$. Let $u$
be a node along the search path with children $v$ and $w$. If our
search path takes a left turn to $v$, then $s_{\max}$ misclassifies
all red points in $w$, so $k_v = k_u + R_w$. Similarly if our search
path takes a right turn to $w$, then $k_w = k_u + B_v$. Once we reach
a leaf $u$ we know that $\Mis(s_{\max}) = k_u$ ($+1$ if $u$ is a blue
point and $s_{\max}$ lies strictly to the right). If
$\Mis(s_{\max}) \leq k$ that means $s_{\max}$ is valid, so we return
it, otherwise the optimum is the closest valid point. We show below
how to find the closest valid point to the left of $s_{\max}$, and we
find the closest one to the right similarly. We return the closest one
of the two.

To find the closest valid point to the left we start by walking back
up the tree until we are at a node $u$ where the search path goes to
the right child and $k_u + M^u_{\rm mis} \leq k$, so $\hat{s}_u$
misclassifies $k$ or fewer points. We then walk down $u$'s left child,
choosing the right child if it contains a separator misclassifying $k$
or fewer points, otherwise choosing the left child. See the arrows in
Figure \ref{fig:kMisMinMaxWalk}. Similarly we find the closest valid
point to the right, and return the one closest to $s_{\max}$. Since
the tree is balanced, traversing this path takes $O(\log n)$ time. We
have thus established the following result:

\begin{figure}[tb]
\centering
\includegraphics{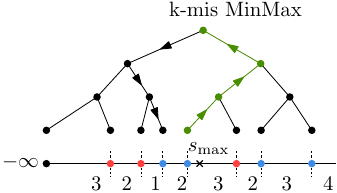}
\caption{The search path towards $s_{\max}$ in green. The walk to find
  the first valid point left of $s_{\max}$ for $k = 1$ is indicated
  with arrows.}
\label{fig:kMisMinMaxWalk}
\end{figure}

\optimalOneD*

\section{Preliminaries}
\label{sec:preliminaries}

\subparagraph{General definitions.} We use the standard point-line
duality that maps any point $p=(p_x,p_y)$ in the primal plane to a
line $p^* : y=p_x x-p_y$ in the dual plane, and any line
$\ell : y=mx + c$ in the primal plane into a point $(m,-c)$ in the
dual plane. 

Let $A$ be a set of $n$ lines, and let $k \in 1..n$. Let the
\emph{lower $\leq k$-level} $L_{\leq k}(A) \subset \R^2$ of $A$ be the set
of points for which there are at most $k$ lines below it. Similarly let the upper $\leq k$-level $L'{\leq k}(A)$ be set of points for which there are at most $k$ lines above it. Let $L_k(A)$ be the
\emph{$k$-level}, the boundary of $L_{\leq k}(A)$. Note that a
$k$-level lies exactly on existing lines in $A$. Although these terms
refer to a region in the plane, with a slight abuse of notation we
will also use them to refer to the part of the arrangement \A of the
lines in $A$ that lies in this region. The complexity of (\A
restricted to) $L_{\leq k}(A)$ is $O(nk)$, and it can be computed in
$O(nk + n \log n)$ time~\cite{lessThanK}. Note that the lower $0$-level
$L_0(A)$ and the upper $0$-level $L'_0(A)$ denote the \emph{lower envelope}
and the \emph{upper envelope} of the set of lines, respectively.

In $O(n \log k)$ time we can compute a \emph{concave chain
  decomposition}~\cite{chanLPViolations,chan16optim_deter_algor_shall_cuttin}
of $L_{\leq k}(A)$: a set of $O(k)$ chains of total complexity $O(n)$
that together cover all edges of \A in $L_{\leq k}(A)$. See Figure
\ref{fig:chainDecomposition}. A convex chain decomposition is defined
similarly.

\begin{figure}[tb]
    \centering
    \includegraphics{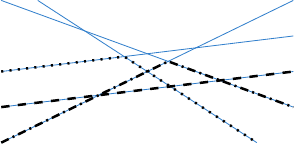}
    \caption{A concave chain decomposition of $L_{\leq 3}(B^*)$, with a dotted, dashed, and dashed-dotted chain.}
    \label{fig:chainDecomposition}
\end{figure}

Throughout this paper we assume points above a separating line $s$
should be blue, and points below should be red. In the dual this means
that lines above separating point $s^*$ should be red, and lines below
should be blue. In particular, we describe algorithms for finding the
optimal separator that classifies in this way. We can then repeat the
algorithm to find the best separator that classifies the other way
around, and finally output the best of the two. For ease of
description we assume all points in $R \cup B$ are in general
position, meaning that all coordinates are unique, and no three points
lie on a line.

\subparagraph{Valid separators.} Fix a value $k \in 1..n$. A separator
$s$ and its dual $s^*$ are \emph{valid} with respect to $k$ if (and
only if) $s \in S_k(B \cup R)$. Line $s$ misclassifies all red points
above $s$ and all blue points below $s$. In the dual, this means all
red lines below $s^*$ and all blue lines above $s^*$ are
misclassified. Consider the dual arrangement of lines $R^* \cup
B^*$. For any two separators $s_1$ and $s_2$ whose duals lie in the
same face of the arrangement, $\Mis(s_1) = \Mis(s_2)$. Let a face
containing valid points be a \emph{valid face}, and note that points
on the boundary of a valid face are also valid. A \emph{valid region}
is the union of a maximal set of adjacent valid faces. Now observe
that $S_k(B \cup R)$ thus corresponds to the union of these valid
regions. With some abuse of notation we use $S_k(B \cup R)$ to
refer to this union of regions in the dual plane as well.

\begin{lemma}[Chan~\cite{chan10bichromatic}]
  \label{lem:k2ValidRegions}
  The set $S_k(B \cup R)$ is contained in
  $L_{\leq k}(R^*) \cap L'_{\leq k}(B^*)$, consists of $O(k^2)$
  valid regions, and its total complexity is
  $O(nk^{1/3}+n^{5/6-\eps}k^{2/3+2\eps}+k^2)$.
\end{lemma}

\begin{lemma}
  \label{lem:k2ValidRegions_lowerbound}
  There may be $\Omega(k^2)$ valid regions of total complexity
  $\Omega(k^2 + ne^{\sqrt{\log k}})$.
\end{lemma}

\begin{proof}
  The $\Omega(ne^{\sqrt{\log k}})$ term follows from
  the fact that we can make an arrangement of red lines whose
  $k$-level has complexity
  $ne^{\Omega({\sqrt{\log
        k}})}$~\cite{toth01point_sets_many_sets}. All blue lines are
  placed sufficiently low so that they do not interfere. The whole red
  $k$-level shows up as boundary of one of the valid regions.

  \begin{figure}
    \centering
    \includegraphics{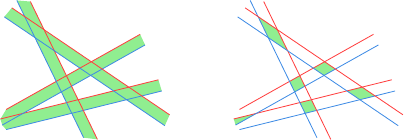}
    \caption{A construction of $q = 4$ red/blue pairs of lines. Left: $k = q-2 = 2$ resulting in $\Theta(k^2)$ valid cells. Right: $k = q-1 = 3$, resulting in a single valid cell of complexity $\Theta(k^2)$.}
    \label{fig:validRegionLowerBound}
  \end{figure}

  Constructions for the $\Omega(k^2)$ term are shown in
  \cref{fig:validRegionLowerBound}. They show that the valid regions
  can have $\Omega(k^2)$ connected components and that a single
  connected component can have complexity $\Omega(k^2)$, as it can
  have this many holes.

  We create $q$ thin strips, each bounded from above by a red line and
  from below by a blue line, such that there are $\Omega(q^2)$ cells
  where two strips intersect. Inside a strip both boundary lines are
  classified correctly, while above or below the strip exactly one of
  them is misclassified. Inside an intersection of two strips we are
  inside two strips and outside $q-2$, so we have $q-2$
  misclassifications. Similarly, inside a single strip we have $q-1$
  misclassifications, and in any 'face' between the strips $q$
  misclassifications. By setting $k = q-2$ only the intersections are
  valid, so we have $\Omega(k^2)$ valid regions of constant
  size.  By setting $k = q-1$ everything inside a strip is valid, and
  we have one large valid region with many holes of complexity
  $\Omega(k^2)$.
\end{proof}

In addition, we observe and will use the following useful properties.

\begin{lemma}
  \label{lem:k2RedBlueIntersects}
  There are $O(k^2)$ red-blue intersections in $L_{\leq k}(R^*) \cap L'_{\leq k}(B^*)$.
\end{lemma}
\begin{proof}
  Consider a concave chain decomposition of $L_{\leq k}(R^*)$
  consisting of $O(k)$ concave chains, and a convex chain
  decomposition of $L'_{\leq k}(B^*)$ consisting of $O(k)$ convex
  chains. Since a concave chain and a convex chain intersect at most
  twice, the $O(k)$ chains intersect at most $O(k^2)$ times.
\end{proof}

\begin{lemma}
  \label{lem:lineKIntersections}
  Any line $\ell$ has $O(k)$ intersections with $L_{\leq k}(R^*) \cap L'_{\leq k}(B^*)$.
\end{lemma}

\begin{proof}
  Line $\ell$ intersects a convex (concave) chain at most twice. Since
  $L_{\leq k}(R^*)$ can be covered by $O(k)$ concave chains, and
  $L'_{\leq k}(B^*)$ can be covered by $O(k)$ convex chains, the
  lemma follows.
\end{proof}

\begin{lemma}
  \label{lem:leftMostRedBlue}
  A valid region $V$ is bounded by red lines on the top and blue lines
  on the bottom. The leftmost point of $V$ is a red-blue intersection,
  or $V$ is unbounded towards the left. The rightmost point in $V$ is
  a red-blue intersection, or $V$ is unbounded to the right.
\end{lemma}
\begin{proof}
  Let $V$ be bounded to the left (otherwise there is nothing to
  prove), and let $s \in V$ be a lowest leftmost point. Therefore, $s$
  must lie on a blue line $b$ (otherwise we could shift it further
  down without increasing the number of misclassifications). Since it
  is a leftmost point, we cannot shift it to the left along $b$ either
  (i.e. it is a local minimum in terms of $x$-coordinates). Therefore,
  it must lie on a second line $\ell$. Since $R^* \cup B^*$ contains
  no vertical lines, it follows that $\ell$ must be red.
\end{proof}

\subsection{Finding a maximum-margin strip}
\label{sub:Finding_a_maximum-margin_strip}

Let $S$ be a maximum margin-strip bounded by parallel lines
$\ell_B$ and $\ell_R$, with $\ell_B$ above $\ell_R$ that separates $B$
and $R$. It is easy to see that $\ell_B$ must contain a blue point and
$\ell_R$ must contain a red point. More precisely, we have:

\begin{lemma}
  \label{lem:maximum_margin_strip}
  Let $b \in \CH(B)$ and $r \in \CH(R)$ be the pair of points
  realizing the Euclidean distance between $\CH(R)$ and
  $\CH(B)$. There is a maximum width strip $S$ separating $B$ and $R$
  bounded by $\ell_B$ and $\ell_R$ such that: (i) $\ell_B$ contains
  $b$, (ii) $\ell_R$ contains $r$, and (iii) $S$ is perpendicular to
  the line segment $\overline{br}$.
\end{lemma}

\begin{proof}
  By Edelsbrunner~\cite[Lemma
  2.1]{edelsbrunner85comput_extrem_distan_between_two_convex_polyg},
  either $b$ or $r$ is a vertex, and the other point lies on an
  edge. Assume that $b$ is a vertex of $\CH(B)$, and $r$ lies on edge
  $e_r$ of $\CH(R)$. The case $r$ is a vertex of $\CH(R)$ is
  symmetric. By definition of the convex hull, $\CH(R)$ lies in one of
  the halfplanes bounded by $\ell_R$. Assume without loss of
  generality it is the halfplane below $\ell_R$. Similarly, it now
  follows by convexity that $\CH(B)$ (and thus $B$) lies above
  $\ell_B$. Hence, the strip $S$ is empty of points in $B \cup R$ and
  separates $B$ and $R$. Since $S$ is perpendicular to
  $\overline{br}$, the width $w$ of this strip equals the length,
  $\dist(b,r)$, of the line segment $\overline{br}$.

  A maximum width separating strip $S^*$ has width
  $w^* \geq w = \dist(b,r)$. However, since such a strip separates $B$
  and $R$ it follows that $\dist(\CH(B),\CH(R))=\dist(b,r)$ is at least
  $w^*$. Hence, it follows that
  $\dist(b,r)=w\leq w^* \leq \dist(b,r)$. Hence $S$ is actually an
  maximum width separating strip.
\end{proof}

By Lemma~\ref{lem:maximum_margin_strip} we can thus compute a maximum
width strip---and thus a separating line that maximizes
$M_{\textit{strip}}$---by finding the pair of points that realizes the
minimum distance between $\CH(B)$ and $\CH(R)$. Given two sets of
points $R$ and $B$, computing the distance between their convex hulls
$\CH(B)$ and $\CH(R)$ is an LP-type
problem~\cite{gartner95subex_algor_abstr_optim_probl}, and can be
solved in linear
time~\cite{chazelle96linear_time_deter_algor_optim}. Hence, we can
find a maximum margin separator in linear time. This result actually
extends to higher dimensions as well.

We can actually maintain a maximum margin separating line under
updates efficiently as well, by maintaining the convex hulls $\CH(B)$
and $\CH(R)$ in the data structure of Overmars and van
Leeuwen~\cite{overmars81maint}. Their data structure uses linear
space, and allows insertions and deletions of points in $O(\log^2 n)$
time. Since their structure explicitly maintains the upper and lower
hull of $\CH(R)$ and $\CH(B)$ in a binary search tree, we can directly
use the algorithm of
Edelsbrunner~\cite{{edelsbrunner85comput_extrem_distan_between_two_convex_polyg}}
to recompute a maximum width strip from scratch in $O(\log n)$ time
after every update.

\begin{theorem}
  \label{thm:2d_maxmargin}
  Let $B \cup R$ be a set of $n$ points in $\R^2$. There is an $O(n)$
  space data structure that maintains an optimal separator $s$ with
  respect to $M_\mathit{strip}$, and supports inserting or deleting a
  point in amortized $O(\log^2 n)$ time.
\end{theorem}

Note that there are data structures to maintain the convex hull that
supports updates in $O(\log^{1+\eps} n)$
time~\cite{chanDynamicHalfplaneReporting}, or even in optimal
$O(\log n)$
time~\cite{brodal02dynam_planar_convex_hull,jacob19dynam_planar_convex_hull_arxiv}. However,
it is unclear whether they can be made to support computing the
distance between $\CH(R)$ and $\CH(B)$ efficiently, as these data
structures do not explicitly maintain the hull. Chan's
structure~\cite{chanDynamicHalfplaneReporting} does support testing
whether $\CH(R)$ and $\CH(B)$ intersect (a subroutine of
Edelsbrunner's algorithm). Hence, it would be interesting to try and
extend his approach to support computing the distance as well.

\section{Dynamic linear programming with violations}
\label{sec:Linear_Programming_with_Violations}

In this section we consider the following problem: given a set of $n$
constraints (halfplanes) $H$ in $\R^2$, an objective function $f$, and
an integer $k$, find a point $p$ that violates at most $k$ constraints
and minimizes $f(p)$. We assume without loss of generality that
$f(p)=p_x$, so we are looking for the leftmost \emph{valid} point,
that is, a point that violates at most $k$ constraints. Chan solves
this problem in $O((n + k^2) \log n)$ time~\cite{chanLPViolations}. In
the same time bounds, his approach can find the minimum number
$k_{\min}$ of constraints violated by any point. We give an overview
of his techniques below, and then show how to make the approach
dynamic. We maintain a leftmost valid point $p$ under semi-online insertions
and deletions of constraints. `Semi-online' means that when a
constraint is inserted we are told when it will be deleted. We first
do so for a given value $k$, and then extend the result to maintain
$k_{\min}$.

The above problem of linear programming with violations is a
generalization of (the dual of) our MinMis problem. A point $p$
violates a constraint (halfplane) $h \in H$ if it lies outside of the
halfplane. Let $R$ be the set of lines bounding lower halfplanes, and
$B$ be the set lines bounding upper halfplanes, and color them red and
blue respectively. Then $p$ violates all blue constraints below, and
all red constraints above, and thus $p$ violates exactly the lines in
$X(p, R \cup B)$, and thus violates $\Mis(p)$ constraints. This means
we can solve the MinMis problem and compute
$s_{\textrm{mis}} = \argmin_s \Mis(s)$ in $O((n + k^2)\log n)$ time.

\subsection{Chan's algorithm}
\label{sec:chansAlgorithm}

Chan considers the decision version of the problem: given an integer
$k$, find the leftmost point that violates at most $k$
constraints. Their algorithm actually generates all local minima that
violate fewer than $k$ constraints as well, so by guessing
$k = \sqrt{n}, 2 \sqrt{n}, 4\sqrt{n} \dots$ we can find the minimum
value $k_{\min}$ for which a valid point exists in the same time
bounds.

We first assume the optimum is bounded, i.e. there are no valid
regions that are unbounded towards the left, and handle this case
later. Then by Lemma \ref{lem:leftMostRedBlue}, the leftmost valid
point in a valid region must be a red-blue intersection, and by Lemma
\ref{lem:k2RedBlueIntersects} there are only $O(k^2)$ of them. We
construct the concave chain decomposition of $L_{\leq k}(R)$ and the
convex chain decomposition of $L'_{\leq k}(B)$ in $O(n \log n)$
time, and compute all their intersections in $O(k^2 \log n)$ time;
this gives us all candidate optima.

Consider a red chain $c_r$, as in Figure
\ref{fig:chansAlgorithm}. Every blue chain $c_b$ defines a (possibly
empty) interval on $c_r$, such that points inside the interval lie
above $c_b$ and points outside the interval lie below $c_b$. The
\emph{blue ply} of a point $p$ on $c_r$ is the number of blue chains above $p$ (and thus the number of violated blue constraints above $p$). This is the number of blue intervals
not containing $p$, and thus the number of intervals ending before $p$
or starting after $p$. By storing the start points (end points) of all blue
intervals in a balanced binary tree we
can thus find the blue ply of any point $p$ on $c_r$ in $O(\log k)$
time. We call this the \emph{chromatic ply data structure} of
$c_r$. The chromatic ply data structure of a blue chain $c_b$ is
defined symmetrically.

For an intersection point $p$ between red chain $c_r$ and blue chain
$c_b$ we can now calculate its $\Mis(p)$ value: query $c_r$ for the blue ply and query $c_b$ for the red ply, both in $O(\log k)$ time, and sum them up. For all $O(k^2)$ red-blue intersections this takes $O(k^2 \log k)$ time. Among them we then find the leftmost valid intersection and return it, if it exists. 

If the optimum was unbounded, then the leftmost segment of one of the chains must be valid. We can check this in $O(k \log k)$ time using the chromatic ply data structures. 

\begin{figure}[tb]
    \centering
    \includegraphics{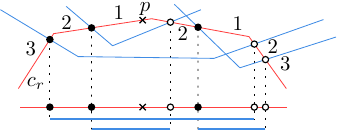}
    \caption{The blue chains create intervals on a red chain $c_r$. We
      can query the blue ply by comparing the number of start points
      (solid) and the number of end points (open).}
    \label{fig:chansAlgorithm}
\end{figure}

The above algorithm takes $O((n + k^2)\log n)$ time. Note that our measure $\Max$ is not linear, and therefor this result is not directly applicable to the $k$-mis MinMax problem.

\subsection{A semi-dynamic data structure for a fixed $k$}

We now make the above algorithm dynamic under semi-online insertions
and deletions: given a value $k$, we maintain the leftmost point that
violates at most $k$ constraints.

We first show how to maintain the concave chain decomposition of $L_{\leq k}(R)$ (and similarly the convex chain decomposition of $L'_{\leq k}(B)$) using an extension of the logarithmic method~\cite{dobkin1989logMethodSemiOnline, smid1990logarithmicExtension}, and then use these chains to actually maintain the leftmost valid separator. 

\subsubsection{Maintaining the concave chain decomposition}
First we maintain the concave chain decomposition of $L_{\leq k}(R)$
using Dobkin and Suri's extension of the logarithmic
method~\cite{dobkin1989logMethodSemiOnline,smid1990logarithmicExtension}. The idea is to
maintain a partition of $R$ into $z = O(\log n)$ subsets
$R_0, R_1.. R_z$, such that for each layer $i$ the following conditions hold:

\begin{description}
\item [(1)] none of the lines in set $R_i$ will be
  deleted for at least $2^i$ updates after the set is
  created.
\item [(2)] $|R_i| = O(2^i)$.
\end{description}

For each set $R_i$ we separately store the concave chain decomposition
of $L_{\leq k}(R_i)$. Since each such structure contains $O(k)$
chains, we have $O(k \log n)$ chains in total. The union of these
chains also covers $L_{\leq k}(R)$: if a line $\ell$ is among the
lowest $k$ lines in $R$ at some $x$-coordinate, it must also be among
the lowest $k$ lines in any subset $R' \subseteq R$, including the
subset $R_i$ containing $\ell$.

The basic idea is the following. After set $R_i$ is created, by
condition (1) no items will be deleted from it for at least $2^i$ updates, so
it remains fixed for $2^i$ updates and gets rebuilt after that. As
such, the smaller data structures are rebuilt quite often, and the
larger data structures remain fixed for a long time. By construction,
deletions happen only at layer $0$ from set $R_0$, which contains very few ($O(1)$)
lines. Lines are inserted in layer $0$, and gradually move to
higher layers where they remain fixed for an ever increasing
number of updates. When a line is to be deleted soon, it
gradually moves down to layer $0$ again. Next, we specify how to initialize
and update this data structure.

\subparagraph{Initialization.} Say we have some initial set of lines
$R$ with $|R| = n$. We put the first line to be deleted ($d = 1$) in
set $R_0$, and then iterate through all remaining lines $r \in R$ with
$d > 1$. If line $r$ is to be deleted after $d$ deletions,we place it at layer $\lceil\log(d)\rceil -1$. This indeed satisfies
condition (1) and (2). Afterwards, we build the concave chain
decomposition of each set $R_i$. This all takes $O(n \log n)$ time.

\subparagraph{Update.} Assume we have to perform $m = 2^z = \Omega(n)$
updates, numbered $0$ through $m-1$ (if $m \gg n$ we can rebuild the
entire data structure every $n$ updates). For every $0 \leq i \leq z$,
we divide the updates into consecutive \emph{blocks} of size $2^i$:
that is, block $j$ at layer $i$ consists of updates $j 2^i$ through
$(j+1)2^i-1$. See Figure \ref{fig:logarithmicMethodChains} for a
schematic illustration.

Suppose we are processing update $u$. We first perform the update on
$R_0$, i.e. insert a new line into $R_0$ or remove a line from
it. Then, we have to ensure that condition (1) and (2) still hold. Let $j(i,u)$ be the block at layer $i$ during update $u$, i.e. $j(i,u) = \lfloor u /2^i \rfloor$. Let
$i'$ be the highest layer for which the current block ends after update $u$ (i.e. $i'$ is the largest integer for which $u = (j(i',u)+1)2^{i'}-1$). That means that, for every layer $y \leq i'$, we have performed exactly $2^y$ updates since set $R_y$ was created, and items in set $R_y$ might be due to be deleted soon: set $R_y$ no longer satisfies condition (1). Therefore, we destroy all data structures at layers $0$ through $i'$, and redistribute all $O(2^{i'})$ items over sets $R_0$ through $R_{i'}$ as we did during initialization; lines that are to be deleted after $d > 2^{i'}$ deletions are put in $2^{i'}$. Afterwards, conditions (1) and (2) hold again.

\begin{figure}
    \centering
    \includegraphics{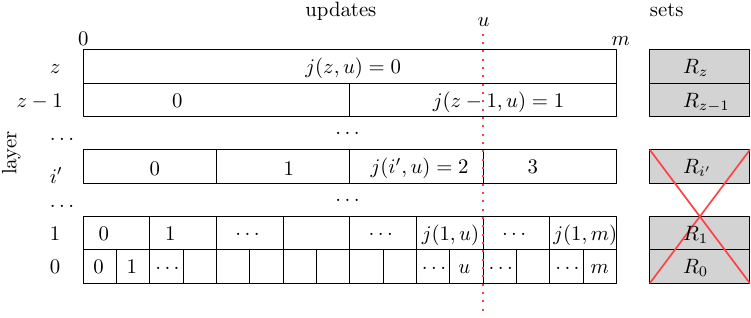}
    \caption{Updates $0$ through $m$ are divided into blocks at every layer. At update $u$, all data structures below layer $i'$ are destroyed and rebuilt.}
    \label{fig:logarithmicMethodChains}
\end{figure}

\subparagraph{Analysis.}  Consider set $R_i$ at layer $i$. This set is
rebuilt every $2^i$ iterations (whenever a block at layer $i$ ends),
after which we have to rebuild the concave chain decomposition of
$L_{\leq k}(R_i)$ in $O(2^i \log 2^i) = O(2^i i)$ time. During all
$2^i-1$ updates in between, this set remains fixed, and we do not have
to do anything. This results in an amortized update time for set $R_i$
of $O(2^i i / 2^i) = O(i)$. Summing this up over all layers results in
$O(\sum_{i = 0}^z i) = O(z^2) = O(\log^2 n)$ amortized update time. As
Smid shows~\cite{smid1990logarithmicExtension}, this amortized bound
can be turned into a worst-case bound.

The chain decomposition of set $R_i$ uses $O(|R_i|)$ space, so the total space usage is $O(n)$.

\begin{lemma}
  \label{lem:semiOnlineChains}
  We can maintain $O(k \log n)$ concave chains of total complexity
  $O(n)$ that cover $L_{\leq k}(R)$ under semi-online insertions and
  deletions in $O(\log^2 n)$ time.
\end{lemma}

In fact, we can maintain a slightly altered version of the above
data structure within the same time and space bounds. Let $2^x$ be the
smallest power of two that is at least $k \log n$, i.e. $2^{x-1} < k
\log n \leq 2^x$. We store the $2^x$ lines that are the first to be
deleted in a separate list, the \emph{leftover} list. We do not build
a chain data structure on the leftover lines, but instead we let each
leftover line form a trivial chain, so we still have $O(k \log n)$ chains covering $L_{\leq k}(R)$. This way all sets $R_j$ with $j < x$ are empty. We can thus perform $2^x$ \emph{cheap} updates without having to modify any of the sets $R_j$: we can simply insert directly in (or delete directly from) the leftover list, without having to rebuild any data structure, in $O(1)$ time. Once every $2^x$ updates we perform an \emph{expensive} update, destroying all sets at layers $0 \dots i'$ as before, and redistributing all those lines and the leftover lines over layers $0$ through $i'$ again. Each set $R_i$ still has an amortized update time of $O(i)$, and thus the total amortized update time remains $O(\log^2 n)$.

\subsubsection{Maintaining intersections and chromatic ply data structures}
Next, we show how to maintain the chromatic ply data structures on the chains, which also gives us the $O(k^2)$ red-blue intersections in $L_{\leq k}(R) \cap L'_{\leq k}(B)$.

We maintain a concave chain decomposition of $L_{\leq k}(R)$ and a convex chain decomposition of $L'_{\leq k}(B)$ using Lemma \ref{lem:semiOnlineChains}, specifically the slightly altered version with $O(k\log n)$ leftover lines. We maintain the two simultaneously with a shared update counter $u$, such that we always perform the same `type' of update (cheap or expensive) on both structures. On each red chain $c_r$ we maintain a blue ply data structure, consisting of a list of startpoints and a list of endpoints of blue intervals induced on $c_r$ by blue chains. Similarly on each blue chain $c_b$ we maintain a red ply data structure. Additionally, we maintain a set $I$ of the $O(k^2 \log^2 n)$ red-blue intersections between the chains; all red-blue intersections in $L_{\leq k}(R) \cap L'_{\leq k}(B)$ are contained in $I$.

Consider the insertion of a line $r$. Depending on the type of update, we do the following:

\begin{description}
    \item[Cheap update:] line $r$ is added to the leftover list, and
      forms a trivial chain $c_r$. We compute all $O(k \log n)$
      intersections between $c_r$ and the blue chains, insert them in $I$, and build the blue ply data structure on $c_r$. For each intersected blue chain $c_b$ we insert the interval induced on $c_b$ by $c_r$ into the red ply data structure of $c_b$. Now $I$ and all chromatic ply data structures are up-to-date again. We handle a deletion of a line similarly. This update takes $O(k \log^2 n)$ time.
    
    \item[Expensive update:] some number of blue and red chains are
      destroyed and rebuilt. It would be difficult to update the
      chromatic ply data structures of all other chains efficiently. So
      instead we destroy all chromatic ply data structures and the set
      $I$, and rebuild them from scratch. Since we have $O(k \log n)$ red and blue chains, recomputing $I$ and the chromatic ply data structures takes $O(k^2 \log^3 n)$ time.
\end{description}

Every $2^x = O(k \log n)$ updates we have $2^x-1$ cheap updates,
taking $O(k \log^2 n)$ time each, and $1$ expensive update, taking
$O(k^2 \log^3 n)$ time. This takes
$O((k\log n) * (k \log^2 n) + 1 * (k^2 \log^3 n)) = O(k^2 \log^3 n)$
time for $2^x$ updates, making the amortized updates
time $O(k \log^2 n)$. We thus have:

\begin{lemma}
  \label{lem:semiOnlineIntersections}
  We can maintain a set $I$ of $O(k^2 \log^2 n)$ bichromatic
  intersection points containining all bichromatic intersections in
  $L_{\leq k}(R) \cap L'_{\leq k}(B)$ under semi-online updates
  in amortized $O(k \log^2 n)$ time. Our data structure uses
  $O(n + k^2 \log^2 n)$ space.
\end{lemma}

\subsubsection{Maintaining the leftmost valid point}

The last step is to maintain the leftmost valid point $s$ for a fixed value $k$. We know $s$ is contained in the set $I$ maintained by Lemma \ref{lem:semiOnlineIntersections}, but simply iterating through the entire set each update would take too long. We wish to build a data structure on the points in $I$ that maintains $\Mis(p)$ for each $p \in I$, and can handle the following operations:

\begin{description}
    \item[Insertion/Deletion:] Inserting or deleting a point (a red-blue intersection).
    \item[Halfplane update:] Correctly update $\Mis(p)$ for each $p \in I$ after the insertion or deletion of a constraint, e.g. increment $\Mis(p)$ by one for all points $p$ in the halfplane above an inserted line $r$ (or in the halfplane below an inserted line $b$).
    \item[Query:] Given a query value $k' \leq k$, return the leftmost point
      $p \in I$ with $\Mis(p) \leq k'$.
\end{description}

We can achieve the above using partition
trees~\cite{chanPartitionTrees}; we first recall their definition. Let
$P$ be a set of $m$ points on which we wish to create a partition
tree, and let $q$ be some large constant (in the literature this
constant is generally called $r$). At the root node, the set of points
$P$ is partitioned into $q$ subsets $P_1 \dots P_q$ using a simplicial
partition (a partition of space into triangles), with the property
that any line intersects $O(\sqrt{q})$ of the triangles. For each
subset $P_i$ the subtree is built recursively, and once a subset
contains a constant number of points we create a leaf. Given a
halfplane $h$, the tree can return a set of $O(\sqrt{m})$ nodes
corresponding to all points of $P$ contained in $h$. The tree uses
$O(m)$ space, can be built in expected $O(m \log m)$ time, and answers queries
in expected $O(\sqrt{m})$ time.

We first consider the somewhat `static' version of this problem where the set $I$ is fixed, and we only have to handle halfplane updates and queries. We store $I$ in a partition tree, and separately in a balanced binary search tree sorted on $x$-coordinate. Each node $u$ of the partition tree contains some set $I(u) \subseteq I$ of points. Let $k_{\min}(u) = \min_{p \in I(u)}\Mis(p)$ be the smallest number of violations achieved by a point in $u$. We would like to maintain $k_{\min}(u)$ in order to answer queries, but doing so explicitly is expensive since a halfplane operation could update this value in all nodes. Therefore we instead maintain two values in each node $u$: a \emph{buffer} value $b(u)$, and a \emph{partial} value $k'_{\min}(u)$, under the invariant that $k_{\min}(u) = k'_{\min}(u) + \sum_{a \in A(u)} b(a)$ where $A(u)$ is the set of ancestors of $u$ (including $u$ itself). We additionally maintain the leftmost point $p_{\min}(u) \in I(u)$ with $\Mis(p_{\min}(u)) = k_{\min}(p_{\min}(u))$.

When we initially build the partition tree on $I$ we explicitly compute the number of violations $\Mis(p)$ for all $p \in I$ (using the chromatic ply data structures on the chains). We then go through the tree in a bottom-up fashion, and for each node $u$ compute $p_{\min}(u)$, set $b(u) = 0$ and set $k'_{\min}(u) = \Mis(p_{\min})$. Clearly, this satisfies the invariant.

Whenever we visit a node $u$, during any operation, we will \emph{propagate the buffer}: we update $k'_{\min}(u) \pluseq b(u)$, then for each child $c$ of $u$ we do $b(c) \pluseq b(u)$, and afterwards we set $b(u) = 0$. Observe that after this operation the invariant $k_{\min}(u) = b(u) + k'_{\min}$ still holds for $u$ and all its children. Specifically for node $u$ itself we know that $k'_{\min}(u) = k_{\min}(u)$, since $b(a) = 0$ for all ancestors $a \in A(u)$.

Consider a halfplane update operation, the insertion of a red line $r$ (blue
lines and deletions are handled analogously). See Figure
\ref{fig:partitionTreeLP}. We start at the root node $u$, and first
propagate its buffer. Then we find all children of $u$ whose triangles
are fully above $r$. For each such child $c$, all points in $I(c)$
violate $r$, and thus $k_{\min}(c)$ increases by one, so we
increment the buffer $b(c)$ by one to satisfy the invariant. For all $O(\sqrt{q})$ children
whose triangles are intersected by $r$, we can afford to
recurse. After all intersected children are recursively updated, we
look at all children $C(u)$ and set
$k'_{\min}(u) = \min_{c \in C(u)} k'_{\min}(c) + b(c)$.

\begin{figure}
    \centering
    \includegraphics{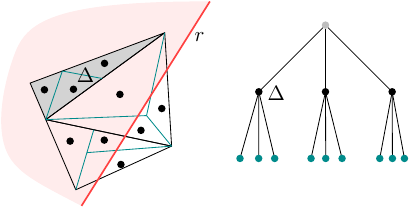}
    \caption{A partition tree of depth two on a set of points $I$ with
      three triangles at the first level. When inserting a line $r$,
      we increment the buffer of $\Delta$, and recurse on the
      intersected triangles.}
    \label{fig:partitionTreeLP}
\end{figure}

Consider now a query, where we wish to find the leftmost point $p$ that is valid for a given value $k'$. For a value $x$, let $w(x)$ be the minimum number of constraints violated by any point $p \in I$ with $p_x \leq x$, and observe that $w(x)$ is monotonically decreasing in $x$. Hence, we can binary search on $x$ to find the smallest $x$-coordinate for which $w(x) \leq k'$. We thus have to solve the decision problem: given a value $k'$ and a value $x$, does there exist a point $p$ with $p_x \leq x$ such that $\Mis(p) \leq k'$? If such a point exists we look for a smaller $x$, otherwise we look for a larger $x$. We can answer this decision problem by querying the partition tree with a vertical halfplane at $x$. Starting at the root $u$, we propagate the buffer. Then for each fully contained child $c$ we compute $k_{\min}(c) = k'_{\min}(c) + b(c)$, and if $k_{\min}(c) \leq k'$ we are done and we return true. Otherwise, we recurse in the children intersected by the vertical query halfplane. 

Lastly we consider the insertion or deletion of points, which we both
handle in standard ways. We handle insertions using the logarithmic
method, creating $O(\log n)$ partition trees. We perform all halfplane
operations and queries on all trees, increasing their runtime by a
factor $\log n$. We perform the deletions implicitly: we can mimic the
deletion of a point $p$ by setting $\Mis(p) = \infty$, such that the
point $p$ can never be returned as a valid point anymore. We walk down
the partition tree towards the leaf containing $p$, propagating the
buffer at every step. Then in bottom-up fashion we update
$k'_{\min}(u)$ and $p_{\min}(u)$ for every node $u$ on the path,
starting in the leaf node, as follows. If $p_{\min}(u) \neq p$ we do
not have to do anything, since removing $p$ does not make any
difference for $k'_{\min}(u)$ or $p_{\min}(u)$. If $p_{\min}(u) = p$
we have to find the new minimum: we look at all children $C$ and find
the child $c' = \argmin_{c \in C} k'_{\min}(c) + b(c)$ containing the
best point, and set $k'_{\min}(u) = k'_{\min}(c')$ and
$p_{\min}(u) = p_{\min}(c')$. Observe that the invariant still
holds. Every $n/2$ deletions we rebuild all partition trees from
scratch to ensure that the trees do not contain too many deleted
points.

\begin{lemma}
\label{lem:semiOnlineIntersectionPartitionTree}
We can build a data structure on a set of $O(k^2 \log^2 n)$ points $I$
that can perform insertions and deletions in $O(\log k \log n)$ time,
halfplane updates in expected $O(k \log^2 n)$ time, and queries in
expected $O(k \log^3 n)$ time. The data structure uses $O(k^2 \log^2 n)$ space.
\end{lemma}
\begin{proof}
Clearly the data structure uses $O(k^2 \log^2 n)$ space, since $|I| = O(k^2 \log^2 n)$ and a partition tree uses linear space.

An insertion takes $O(\log k \log n)$ time by the logarithmic method, and a deletion takes $O(\log k)$ time since we only traverse a single path to a leaf. A halfplane update takes $O(\sqrt{k^2 \log^2 n}) = O(k \log n)$ time per tree, and thus $O(k \log^2 n)$ time in total: we spend constant time in each node and recursively visit $O(\sqrt{q})$ children, and it is well known that $Q(m) = O(\sqrt{q}) Q(m/q) + O(1) = O(\sqrt{m})$ for constant fan-out $q$. Similarly for a query each decision problem takes $O(k \log n)$ time to solve, resulting in $O(k \log^2 n)$ time per tree to find the leftmost point $p$ using binary search, resulting in $O(k \log^3 n)$ total time.
\end{proof}

We can now dynamically maintain the solution to an LP with at most $k$ violations. We maintain the chain decompositions of $L_{\leq k}(R)$ and $L'_{\leq k}(B)$ using Lemma \ref{lem:semiOnlineChains}, and the set of their intersections $I$ using Lemma \ref{lem:semiOnlineIntersections}. We additionally build the data structure from Lemma \ref{lem:semiOnlineIntersectionPartitionTree} on the set $I$, and update it as follows.

Each cheap update, e.g. the insertion of a line $r$, we find the $O(k \log n)$ intersections between $r$ and blue chains, and insert them in $O(k \log^2 n \log k)$ total time. We then perform one halfplane update to update the $k'_{\min}(u)$ and $b(u)$ values, in $O(k \log^2 n)$ time.

Each expensive update we discard the data structures from Lemma \ref{lem:semiOnlineIntersections} and \ref{lem:semiOnlineIntersectionPartitionTree} and rebuild them from scratch. This takes $O(k^2 \log^3 n)$ time, and thus $O(k \log^2 n)$ amortized time.

After every update we perform one query with $k' = k$ in $O(k \log^3 n)$ time. Thus:

\linearProgramming*

\subsection{Dynamically maintaining \texorpdfstring{$k_{\min}$}{kmin}}

With the data structure from Theorem \ref{thm:linear_programming} we
can thus maintain the leftmost valid point for a fixed value
$k$. However, Chan's static algorithm could also find the smallest
value $k_{\min}$ for which there exists a valid point, and a leftmost
point $p_{\min}$ with $\Mis(p_{\min}) = k_{\min}$. The partition tree
from Lemma \ref{lem:semiOnlineIntersectionPartitionTree} can handle
such queries as long as $k_{\min} \leq k$; in fact, this value and the
point $p_{\min}$ are stored in the root. However when insertions cause
$k_{\min}$ to increase to $k_{\min} > k$, then $p_{\min}$ no longer
has to lie in $L_{\leq k}(R) \cup L'_{\leq k}(B)$, and thus $p_{\min}$
is not necessarily a member of the set of intersection $I$ maintained
by Lemma \ref{lem:semiOnlineIntersections}.

We can solve this issue by slightly adapting Lemma \ref{lem:semiOnlineChains}. When the set $R_i$ is created, we used to build the concave chain decomposition of $L_{\leq k}(R_i)$ in $O(|R_i| \log |R_i|)$ time (note this time is independent from $k$). Instead, we build multiple chain decompositions: for $l \in 1, 2 \dots \lceil \log n \rceil$ we build the chain decomposition of $L_{\leq 2^l}(R_i)$, in total time $O(|R_i| \log^2 |R_i|)$. Only one of these chain decompositions can be \emph{active} at a time, the others are \emph{inactive}; below we choose the active decomposition such that the union of the active chains form a concave decomposition of $L_{\leq k_{\min}}(R)$.

At initialization, we compute the value $k_{\min}$ using Chan's static algorithm. We decide that the next $2^x - 1$ updates, where $2^{x-1} < k_{\min} \leq 2^x$, will be cheap updates where we do not change any sets $R_i$, and in $2^x$ updates we will perform an expensive update. Since $k_{\min}$ can only increase by at most one each update, this ensures that $k_{\min} \leq 2^{x+1}$ until the next expensive update. We build the data structures from Theorem \ref{thm:linear_programming} for $k = 2^{x+1}$, with the above adaptation that for each set $R_i$ we build multiple chain decompositions. For each set $R_i$ we activate the chain decomposition of $L_{\leq 2^{x+1}}(R_i)$, and deactivate the other chain decompositions, ensuring that during the next $2^x - 1$ cheap updates the optimum $p_{\min}$ will lie inside the active chain decomposition.

For the next $2^x - 1$ cheap updates, we update the data structures as normal. Update $2^x$ is an expensive update again. We query the partition trees for the current optimal value $k_{\min}$, and then destroy the chromatic ply data structures and partition trees. As during initialization, we decide that the next $2^x - 1$ updates, where $2^{x-1} < k_{\min} \leq 2^x$, will be cheap updates. We update the sets $R_i$ as normal, activate the chain decomposition of $L_{\leq 2^{x+1}}(R_i)$ for each set, and rebuild the chromatic ply data structures and partition trees from scratch for $k = 2^{x+1}$. After each update we can simply query the partition tree for $k_{\min}$ and $p_{\min}$.

Both the cheap and expensive updates still have the same asymptotic runtime (with $k$ replaced with $O(k_{\min})$), but an expensive update now occurs every $O(k_{\min})$ updates rather then every $k \log n$ updates. This results in $O(k_{\min} \log^3 n)$ expected amortized update time.

\begin{lemma}
  \label{lem:maintain_kmin}
Let $H$ be a set of $n$ halfspaces in $\R^2$, and let $f$ be a linear objective function. There is an $O(n \log n + k_{\min}^2 \log^2 n)$ space data structure that maintains a point $p$ that violates as few constraints $k_{\min}$ as possible from $H$ and minimizes $f$, and supports semi-online updates in expected amortized $O(k_{\min} \log^3 n)$ time.
\end{lemma}

\section{Exact algorithms for \texorpdfstring{$k$}{k}-mis MinMax}
\label{sec:An_exact_algorithm_in_2d}

For the $k$-mis MinMax problem we are given point sets $R$ and $B$ and
an integer $k$ and wish to compute a separator
$\sopt = \argmin_{s \in S_k(R \cup B)} \Max(s)$ that misclassifies at
most $k$ points and minimizes the distance to the furthest
misclassified point. In this section we present an exact algorithm for
this problem. In Section~\ref{sub:exact_properties} we first discuss
some geometric properties of an optimal solution. In
Section~\ref{sub:computing_valid_regions} we then present algorithms
to construct the valid regions $S_k(R \cup B)$. Finally, in
Section~\ref{sub:computing_opt}, we show how we can then compute an
optimal separator.

\subsection{Geometric properties}
\label{sub:exact_properties}

\subparagraph{Properties of an optimal MinMax solution.} We first
consider the (easier) MinMax problem. Here, we wish to compute
$s_{\max} = \argmin_s \Max(s)$, a separator with minimal distance to
the farthest misclassified point.

\begin{figure}[b]
\centering
\includegraphics{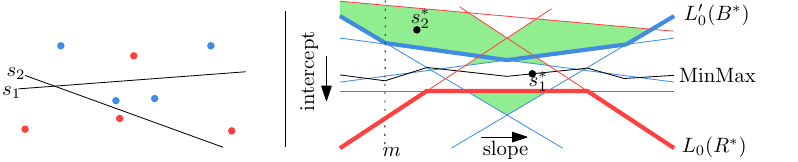}
\caption{Some primal points (left) with their dual (right). Valid
  cells for $k = 2$ are green.}
 \label{fig:baseDual}
\end{figure}

First we consider the problem for lines with a fixed slope $m$; note
this is a vertical line $x=m$ in the dual plane. For slope $m$ there
is some optimal intercept $c$ such that the resulting separating line
$s_{\max}(m): mx + c$ minimizes the error $\Max(s_{\max}(m))$. The
error $\Max(s_{\max}(m))$ is defined by a red point $r$ and a blue
point $b$.  These are misclassified points with the largest distance
to $s_{\max}(m)$, the extremal points, as shown in
\cref{fig:givenSlope}. Separator $s_{\max}(m)$ minimizes distance to
these points, so it is in their middle with equal distance to both. In
the dual plane, clearly $r^*$ must lie on $L_0(R^*)$ and $b^*$ must
lie on $L'_0(B^*)$.

\begin{figure}[tb]
    \centering
    \includegraphics{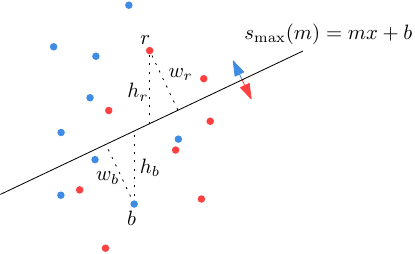}
    \caption{The optimal separating line for a fixed slope $m$, with
      extremal points $r$ and $b$.
    }
    \label{fig:givenSlope}
\end{figure}
\begin{restatable}{lemma}{MinMax}
\label{lem:MinMax}
In the dual plane, for a fixed slope $m$, the optimal MinMax separator $s_{\max}^*(m)$ is vertically in the middle of $L_0(R^*)$ and $L'_0(B^*)$ at $x = m$.
\end{restatable}
\begin{proof}
Let $r^*$ be an extremal red line in $L_0(R^*)$ for slope $m$, and $b^*$ an extremal blue line in $L'_0(B^*)$ for slope $m$. We wish to show that point $s_{\max}^*(m)$ is vertically in the middle of lines $r^*$ and $b^*$.
Since vertical distance is preserved when dualizing, it is equivalent to show in the primal that line $s_{\max}(m)$ is vertically in the middle of points $r$ and $b$. See Figure \ref{fig:givenSlope}. 

Let $\ell$ be a line with slope $m$ and $p$ be a point, and let $w_p$ be the Euclidean distance between $\ell$ and $p$ (the width of the strip), and let $h_p$ be the vertical distance (the height of the strip). Let $\rho(m) = \frac{h_p}{w_p}$ be the ratio between vertical and Euclidean distance, and note that $\rho(m)$ only depends on $m$, not on the specific choice of $\ell$ and $p$.

Since $s_{\max}(m)$ is optimal it must have equal (Euclidean) distance to
both extremal points, so $w_r = w_b$. Because of the fixed ratio
$\rho(m) = \frac{w_r}{h_r} = \frac{w_b}{h_b}$, this means
$h_r = h_b$. That means $s_{\max}(m)$ is vertically in the middle of $b$
and $r$, and thus $s_{\max}^*(m)$ vertically in the middle of $L_0(R^*)$ and $L'_0(B^*)$.
\end{proof}

We now generalize the above solution to any slope. Consider the curve
$\{s^*_{\max}(m) \mid m \in \R \}$ in the dual plane representing the
optimal separator for the MinMax problem for any given slope, shown in
black in \cref{fig:baseDual}. In the rest of this article we will
refer to this as the \emph{MinMax curve}, or simply MinMax. By Lemma
\ref{lem:MinMax} the MinMax curve is in the middle of the red and blue
envelopes $L_0(R^*)$ and $L'_0(B^*)$ for every slope $m$. Since
$L_0(R^*)$ and $L'_0(B^*)$ are both polygonal chains with $O(n)$
vertices, MinMax is also a polygonal chain with $O(n)$ vertices.

\begin{restatable}{lemma}{minMaxTowardsVertexBetter}
\label{lem:minMaxTowardsVertexBetter}
Let $s$ be a point in the interior of an edge $e$ of the MinMax curve. Either moving $s$ left or moving $s$ right along $e$ will decrease the error $\Max(s)$.
\end{restatable}

\begin{proof}
Since $s$ lies on the interior of a MinMax edge, $\Max(s)$ is defined by the distance to exactly one extremal red point $r$ and one extremal blue point $b$ in the primal plane as in \cref{fig:rotate}. Let $m$ be the midpoint in between $b$ and $r$, and note that $m$ lies on $s^*$. Let $C_b$ and $C_r$ be circles tangent to $s^*$ centered at $b$ and $r$, respectively. Note that $C_b$ and $C_r$ have radius $\Max(s)$, since $r$ and $b$ are extremal points. 

Points $b$ and $r$ are on opposite sides of $s^*$, with $m$ in between them. Either rotating $s^*$ clockwise or counterclockwise around $m$ will make it intersect both $C_b$ and $C_r$, which means the distance from $s^*$ to $b$ and $r$ decreases, and therefore $\Max(s_{\max})$ decreases. This corresponds to moving $s$ left or right along $e$.
\end{proof}

\begin{figure}[tb]
\centering
    \includegraphics{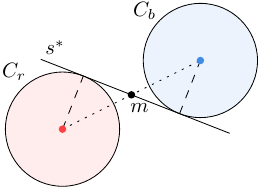}
    \caption{Separator $s^*$ with extremal points $r$ and $b$. Rotating $s^*$ counterclockwise around $m$ will decrease $\Max(s^*)$.}
    \label{fig:rotate}
\end{figure}

\subparagraph{Properties of an optimal $k$-min MinMax solution.} For
the $k$-mis MinMax problem, want to compute an optimal separator
$\sopt = \argmin_{s \in S_k(B \cup R)} \Max(s)$, so a valid separator
with minimal $\Max(\sopt)$. For a fixed slope $m$, let $s_{\opt}(m)$
be an optimal separator with that slope $m$. Lemma~\ref{lem:vert}
below then characterizes such an optimal separator for slope $m$. This
then leads to a characterization of $\sopt$.

\begin{restatable}{lemma}{vert}
  \label{lem:vert}
  Dual point $s_{\opt}^*(m)$ is the valid point with the smallest
  vertical distance to $s_{\max}^*(m)$.

  If there are no valid points $s^*$ with $s^*_x = m$, then $s_{\opt}^*(m)$ does not exist.
\end{restatable}
\begin{proof}
Let $s^*$ be a dual point on the line $x = m$, , and w.l.o.g. assume it lies below $s_{\opt}^*(m)$. Then the blue line $b^*$ on $L'_{\leq k}(B^*)$ at $x = m$ is the furthest misclassified line. As in the proof of Lemma \ref{lem:MinMax}, we know there is a fixed ratio $\rho(m)$ between Euclidean distance $\dist(s, b)$ in the primal and vertical distance between $s^*$ and $b^*$ on the line $x = m$ in the dual. Therefore, minimizing $\dist(s,b)$ is equivalent to minimizing the vertical distance from $s^*$ to $b^*$. It follows that $s_{\opt}^*(m)$ must be the highest valid point below $s_{\max}^*(m)$, or similarly the lowest valid point above $s_{\max}^*(m)$, whichever one is closer.
\end{proof}

\begin{figure}
\centering
\includegraphics{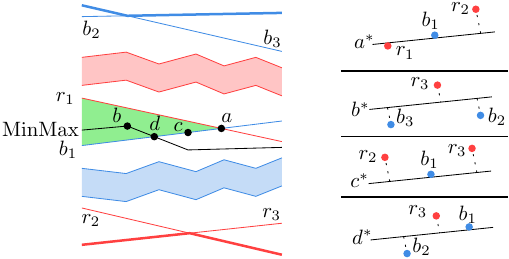}
\caption{Left: the cases \enumit{a, b, c, d} for $s^*_{\opt}$ in the dual plane; the red/blue regions represent some number of correctly classified lines. Right: the cases \enumit{a, b, c, d} for $s_{\opt}$ in the primal plane. }
\label{fig:cases}
\end{figure}

\begin{lemma}
  \label{lem:optimum}
  A point $s^*_{\opt}$ dual to an optimal separator is one of the following:
\begin{enumerate}[label=\enumit{\alph*.}]
    \item A vertex of a valid cell, vertically closest to MinMax.
        
    \item A (valid) vertex of MinMax.
    
    \item The first valid point directly above or below a vertex of MinMax.
    
    \item The intersection of a MinMax edge $e$ with a valid edge, closest to one of $e$'s endpoints.
\end{enumerate}
\end{lemma}

\begin{proof}
Intuition: in the primal plane, the optimal separating line has to be 'bounded' by at least three points. These can either be extremal points that we want the separator to be as close to as possible, or points that the separator is not allowed to cross because that would make it invalid. The four cases are shown in \cref{fig:cases}.
    
Proof by contradiction. Assume the optimum $s^*_{\opt}$ is not any of the four cases \enumit{a, b, c, d}. Then it has to be one of the following cases, for each of which we show that we can change $s^*_{\opt}$ slightly to decrease $\Max(s^*_{\opt})$, meaning $s^*_{\opt}$ was not optimal, which is a contradiction.

\begin{enumerate}

\item $s^*_{\opt}$ is not inside a valid face; clearly this is not possible.
    
\item $s^*_{\opt}$ is strictly inside a valid face not on MinMax. Then moving vertically towards MinMax will decrease $\Max(s^*_{\opt})$ by Lemma \ref{lem:vert}.

\item $s^*_{\opt}$ lies strictly inside a valid face and strictly inside an edge of MinMax. Then by Lemma \ref{lem:minMaxTowardsVertexBetter}, moving towards one of the two adjacent MinMax vertices will decrease $\Max(s^*_{\opt})$.

\item $s^*_{\opt}$ lies strictly inside a valid edge $e^*$, not on MinMax and not above/below a MinMax vertex. 
This means primal line $\sopt$ goes through point $e$, as in \cref{fig:rotate2}. 
It has a single extremal point (if multiple points of the same color would be equally far away we would be below a MinMax vertex, and if a blue and red point would be equally far away we would be on MinMax). 
W.l.o.g. let the extremal point be a blue point $b$, and let $c_b$ be a circle centered at $b$ tangent to $\sopt$. 
We can rotate $\sopt$ either clockwise or counterclockwise around $e$ (clockwise in \cref{fig:rotate2})  to make it intersect $c_b$, decreasing the distance to $b$ and decreasing $\Max(s^*_{\opt})$.

\item $s^*_{\opt}$ lies strictly inside a valid edge $e^*$ above or below a MinMax vertex, or on a valid vertex, but there is another point $s'^*$ with the same $x$-coordinate that is closer to MinMax. Then $\Max(s'^*) < \Max(s^*_{\opt})$ by Lemma \ref{lem:vert}. 

\item $s^*_{\opt}$ lies on an intersection of a MinMax edge $e$ with a valid edge, but $e$ intersects another valid edge in point $s_1$ left of $s^*_{\opt}$ and in point $s_2$ right of $s^*_{\opt}$. Then $s_1$ is closer to the left vertex of $e$, and $s_2$ is closer to the right vertex of $e$, so by Lemma \ref{lem:minMaxTowardsVertexBetter} we have $\Max(s_1) < \Max(s^*_{\opt})$ and $\Max(s_2) < \Max(s^*_{\opt})$.

\end{enumerate}

Now the only remaining cases are cases \enumit{a, b, c, d}, proving the lemma.
\end{proof}

\begin{figure}
\centering
    \includegraphics{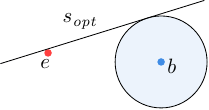}
    \caption{Separator $s'^*$ going through point $e$ with extremal point $b$.}
    \label{fig:rotate2}
\end{figure}

\begin{restatable}{lemma}{complexity}
  There are $O(nk^{1/3}+n^{5/6-\eps}k^{2/3+2\eps}+k^2)$ points of type
  \enumit{a}, and there are $O(n)$ points of types \enumit{b},
  \enumit{c}, and \enumit{d}.  
\end{restatable}
\begin{proof}
\begin{enumerate}[label=\enumit{\alph*.}]
\item Each type \enumit{a} point is a vertex of $S_k(B\cup R)$. By
  Lemma~\ref{lem:k2ValidRegions} the total complexity of
  $S_k(B\cup R)$, and thus also the number of type \enumit{a} points,
  is $O(nk^{1/3}+n^{5/6-\eps}k^{2/3+2\eps}+k^2)$.

  \item MinMax has $O(n)$ vertices.

  \item Again, MinMax has $O(n)$ vertices. Since there are at most two type $c$ points per MinMax vertex, the closest valid one above and below, there are also $O(n)$ type $c$ points.

  \item Each edge $e$ of MinMax has at most one point closest to its left endpoint, and one point closest to its right endpoint, and MinMax has $O(n)$ edges.
      \qedhere
\end{enumerate}
\end{proof}

\subsection{Constructing the valid regions}
\label{sub:computing_valid_regions}

We present two general algorithms for constructing $S_k(B \cup R)$,
i.e.\ the union of valid regions. By Lemma~\ref{lem:k2ValidRegions}
all valid points lie inside
$L_{\leq k}(R^*) \cap L'_{\leq k}(B^*)$, so we present a simple
algorithm that simply constructs this part of the arrangement, and
prunes all invalid regions. This gives an $O(n\log n + kn)$ time
algorithm. We then present a much more involved
$O((nk^{1/3}+n^{5/6-\eps}k^{2/3+2\eps}+k^3)\log^2 n)$ time
algorithm. Finally, if we care only about the case where
$k = k_{min} = \Mis(s_{\textrm{mis}})$, i.e. where we are only allowed
to misclassify as few points as possible (we say the value of $k$ is
\emph{tight} in this case), we can compute $S_k(B \cup R)$ in
$O(k^{4/3} n^{2/3} \log^{2/3} (n / k) + (n + k^2) \log n)$ time.

\subparagraph{A simple algorithm.} By Lemma~\ref{lem:k2ValidRegions}
all valid points lie inside
$L_{\leq k}(R^*) \cap L'_{\leq k}(B^*)$. Consider the overlay
$\A_k(R^*,B^*)$ of $L_{\leq k}(R^*)$ and $L'_{\leq k}(B^*)$.

\begin{lemma}
  \label{lem:slabsComplexity}
  The complexity of $\A_k(R^*,B^*)$ is $O(nk)$.
\end{lemma}
\begin{proof}
  Arrangements $L_{\leq k}(R^*)$ and
  $L'_{\leq k}(B^*)$ both have complexity $O(nk)$, and there are
  only $O(k^2)$ bichromatic intersections
  (Lemma~\ref{lem:k2RedBlueIntersects}), their overlay also has
  complexity $O(nk + k^2) = O(nk)$. 
\end{proof}

\begin{figure}
    \centering
    \includegraphics{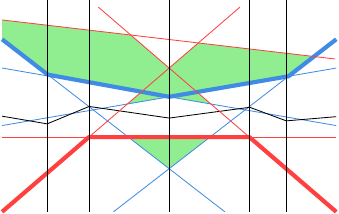}
    \caption{The overlay $\A_k(R^*,B^*)$ has complexity is
      $O(nk)$.} 
    \label{fig:verticalSlabs}
\end{figure}

\begin{lemma}
  \label{lem:construct_simple}
  We can construct $\A_k(R^*,B^*)$ and filter out the valid regions in
  $S_k(R \cup B)$ in $O(n \log n + nk)$ time.
\end{lemma}

\begin{proof}
  We use the following steps:
  \begin{enumerate}
  \item Construct $L_{\leq k}(R^*)$ and $L'_{\leq k}(B^*)$. This can
    be done in $O(nk + n \log n)$ time~\cite{lessThanK}.

  \item Overlay $L_{\leq k}(R^*)$ and $L'_{\leq k}(B^*)$ to get
    $\A_k(R^*,B^*)$. Except for the upper unbounded face which has the
    $k$-level $L_k(R^*)$ as boundary, all faces in $L_{\leq k}(R^*)$ are
    convex. This unbounded face can be trivially triangulated, making
    all faces convex. The same holds for $L'_{\leq k}(B^*)$. We can
    then use Finke and Hinrichs' algorithm~\cite{mapoverlay} to overlay
    two convex subdivisions in time linear in both input and output
    size, which is $O(nk)$ in our case.

  \item Walk through the faces of $\A_k(R^*,B^*)$ in depth-first search
    order, maintaining the number of misclassifications per face, and
    storing it for each face. For two neighboring faces $F_1$ and $F_2$,
    we have $|\Mis(F_1) - \Mis(F_2)| = 1$, since we cross only a single
    line. This means we can maintain the number of misclassifications in
    constant time per step, so going through the whole subdivision takes
    $O(nk)$ time.\qedhere
  \end{enumerate}
\end{proof}

\subparagraph{An output sensitive algorithm.}
Chan~\cite{chan10bichromatic} sketches an approach to compute the
valid region in an output sensitive manner by first computing the
bichromatic intersection points of
$L_{\leq k}(R^*) \cap L'_{\leq k}(B^*)$, and then tracing
$S_k(B\cup R)$, starting from these bichromatic intersection points
``as in a standard $k$-level algorithm''. He states this results in a
running time of $O(|S_k(B\cup R)|\polylog n)$ time, but does not
provide any details. The $k$-level in an arrangement of lines is
connected, whereas here $S_k(B\cup R)$ may consist of $\Omega(k^2)$
disconnected pieces (Lemma~\ref{lem:k2ValidRegions_lowerbound}). This
unfortunately provides some additional difficulties in initializing
the data structure used in the tracing process. Hence, it is not clear
that can indeed be done in $O(|S_k(B\cup R)|\polylog n)$
time. Instead, we present an algorithm that runs in
$O((|S_k(B\cup R)|+n+k^3)\log^2 n)$ time.

\begin{lemma}
  \label{lem:construct_valid_regions_outputsensitive}
  The region $S_k(B\cup R)$ can be constructed in
  $O((n+|S_k(B\cup R)|+n+k^3)\log^2 n)$ time.
\end{lemma}

\begin{proof}
  We first compute the set $Q$ of $O(k^2)$ bichromatic intersection
  points in $L_{\leq k}(R^*) \cap L'_{\leq k}(B^*)$ using the
  algorithm of Chan~\cite{chanLPViolations} (i.e. by computing the
  $O(k)$ concave chains covering $L_{\leq k}(R^*)$ and the convex
  chains covering $L'_{\leq k}(B^*)$, and intersecting them). For
  each such intersection point $v \in Q$, let $R_v \subseteq R^*$
  denote the at most $k$ red lines below $v$, and $B_v \subseteq B^*$
  be the at most $k$ blue lines above $v$. These are the lines
  misclassified by $v$. We wish to trace the boundaries of the valid
  regions of $S_k(B\cup R)$ containing $v$, i.e. for
  which $|B_v|+|R_v|=k$.
  
  Let $F_0 = \{r^- \mid r \in R^*\} \cup \{ b^+ \mid b \in B^* \}$ be
  the set consisting of all halfplanes that are bounded from above by
  a red line, all halfplanes bounded from below by a blue line. We
  store them in the dynamic common halfplane intersection data
  structure of Overmars and van Leeuwen~\cite[Theorem
  7.1]{overmars81maint}.

  Consider a bichromatic intersection point $v \in Q$ for which
  $|B_v|+|R_v|=k$ that we have not traced yet, and at which lines
  $r_v \in R$ and $b_v \in B$ intersect. Consider the case that $b_v$
  lies below $r_v$ to the right of $v$, and thus $v$ is a point on the boundary of a
  valid region $V$. See Figure~\ref{fig:tracing} for an
  illustration. The case $b$ lies above $r$ can be handled
  analogously.

  \begin{figure}[tb]
    \centering
    \includegraphics{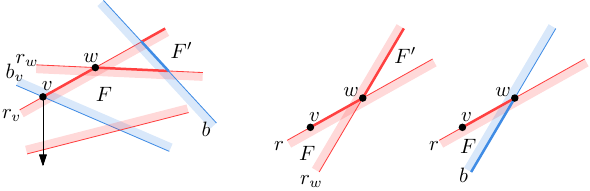}
    \caption{(a) Tracing a valid region $V$. The boundary is shown in
      bold. (b) If the halfplane $r_w^- \in F$, then at the
      intersection of $r$ and $r_w$ the tracing switches to tracing a
      different face $F'$ that also misclassifies $k$ lines. The same
      happens in figure (a) when $r_w$ intersects a blue line $b$ with
      $b^- \in F$. (c) If the halfplane $b^+ \in F$, we keep tracing
      $F$.  }
    \label{fig:tracing}
  \end{figure}

  Before we start tracing $V$, we delete the halfplanes corresponding
  to the lines in $R_v$ from the common halfplane intersection data,
  and replace each such halfplane $r^-$ by $r^+$. Similarly, we
  replace the halfplanes corresponding to the blue lines in $B_v$ with
  $b^-$. Let $F$ denote the current set of halfplanes in the data
  structure. Observe that $v$ is the leftmost point of a (valid) face
  of the arrangement $R^* \cup B^*$; namely the face that is the
  common intersection of the halfplanes in $F$.

  We now trace the boundary of $V$ by repeatedly computing and
  reporting the next clockwise neighbor $w$ of $v$ in $F$. We describe
  what happens when the edge $\overline{vw}$ lies on a red line $r$
  (as the case initially); when the edge $\overline{vw}$ is blue, we
  handle the situation analogously. If line $r$ intersects an other
  red line $r_w$ with $r_w^- \in F$, we continue the tracing from
  $w$. If $r_w^+ \in F$, we replace $r^-$ by $r^+$ and $r_w^+$ by
  $r^-_w$ in $F$, and continue tracing along $r_w$. In this case, we
  moved into another face $F'$ of the valid region, see
  Figure~\ref{fig:tracing}(b). Similarly, if $w$ lies on a blue
  line $b$ with $b^- \in F$, we ``swap'' to the halfplanes $r^-$ and
  $b^+$ as we switch tracing another face. If $w$ lies on a blue line
  $b$ with $b^+ \in F$, we keep tracing in face $F$.

  We repeat the above process until we are back at the initial vertex
  $v$ we started from. We can then again delete and reinsert the
  halfplanes corresponding to the lines in $R_v$ and $B_v$, so that
  the data structure is in its initial form, representing the set of
  halfplanes $F_0$ again.
  
  Computing the $O(k^2)$ bichromatic intersection points takes
  $O((n+k^2)\log n)$
  time~\cite{chanLPViolations,chan16optim_deter_algor_shall_cuttin}. In
  an additional $O(k^3)$ time, we can report all sets $R_v$ and $B_v$
  using the chain decompositions. Initializing the Overmars and van
  Leeuwen structure to represent the common intersection of $F_0$
  takes $O(n\log^2 n)$ time. Every update takes an additional
  $O(\log^2 n)$ time. To trace the boundaries of all valid regions, we
  thus spend $O(|S_k(B \cup R)|)\log^2 n)$ time, and an additional
  $O(k^3\log^2 n)$ time to convert the data structure between $F_0$
  and the halfplanes representing the face containing a point in $V$.
\end{proof}

\subparagraph{An algorithm for the tight $k$ case.} We start with the
following useful observation:

\begin{observation}
  \label{obs:tight_k_valid_regions_are_faces}
  If $k$ is tight, any valid region consists of only a single face.
\end{observation}
\begin{proof}
  Two adjacent faces can not both misclassify exactly $k$ points, since they classify the line dividing them differently.
\end{proof}

By Lemma~\ref{lem:k2ValidRegions} there are $O(k^2)$ valid regions, so
now there are $O(k^2)$ valid faces. Clarkson et
al.~\cite{clarkson1990manyfacesbound} show that $m$ faces in an
arrangement have a complexity of $O(m^{2/3} n^{2/3}+n)$, so our
$O(k^2)$ valid faces have complexity $O(k^{4/3} n^{2/3}+n)$. We then compute the $O(k^2)$ bichromatic intersection points using Chan's algorithm as in
Lemma~\ref{lem:construct_valid_regions_outputsensitive}, but then directly
use Wang's algorithm~\cite{wang2022constructingManyFaces} to construct
all valid faces in
$O(k^{4/3} n^{2/3} \log^{2/3} (n / k) + (n + k^2) \log n)$ time.

\begin{lemma}
  \label{lem:construct_valid_regions_tightk}
  For $k=k_{\min}$, the region $S_k(B\cup R)$ has complexity
  $O(k^{4/3} n^{2/3}+n)$ and can be constructed in
  $O(k^{4/3} n^{2/3} \log^{2/3} (n / k) + (n + k^2) \log n)$ time.
\end{lemma}

\subsection{An algorithm for solving the \texorpdfstring{$k$}{k}-mis
  MinMax problem}
\label{sub:computing_opt}

We now show how, given the valid regions, we can compute an optimal
separator $\sopt \in S_k(B\cup R)$ efficiently. We start by
constructing $L_0(R)$ and $L'_0(B)$, and simultaneously scan through
them to construct the MinMax curve $s^*_{\max}$. This takes
$O(n \log n)$ time~\cite{bookBerg}. By Lemma \ref{lem:optimum} an
optimal separator is of type $a,b,c,$ or $d$. So, we will now compute
all these candidate optima, and iterate through them to find the one
with lowest error.

\subparagraph{Type a points.} Since we are given $S_k(B\cup R)$, we
can simply scan through its vertices, keeping track of the vertex with
the smallest error. To calculate the error of a vertex, we need to
know which segment of $s^*_{\max}$ it lies above/below; then the error
can be calculated in $O(1)$ time. We can compute this in $O(\log n)$
time per vertex using binary search (since MinMax is
$x$-monotone). Hence, this step takes $O(|S_k(B\cup R)|\log n)$ time.

\subparagraph{Type b and c points.} Recall type $b$ points are MinMax
vertices, and type $c$ points are the first valid points above or
below MinMax vertices. We construct the \emph{trapezoidal
  decomposition} of $S_k(B\cup R)$ in $O(|S_k(B\cup R)|\log n)$ time,
so that we can support $O(\log n)$ time point location
queries~\cite{sarnak1986planarLocation}. Each trapezoidal cell has
vertical left and right sides, and up to four neighboring cells each, see Figure
\ref{fig:staticBCD}.

For each vertex of MinMax we perform one point location query, which
tells us what trapezoidal cell the vertex lies in. If this cell is
inside a valid region, the vertex is a type $b$ point. Otherwise the
closest valid edges vertically above and below this vertex are simply
the edges bounding that trapezoid, giving us up to two type $c$
points. Since MinMax has $O(n)$ vertices, this gives us all type $b$
and $c$ points in $O(n \log n)$ time. Including the time to build the
decomposition, this thus takes $O(|S_k(B\cup R)| + n) \log n)$ time.

\begin{figure}[tb]
    \centering
    \includegraphics{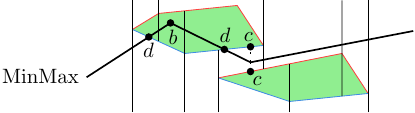}
    \caption{Two valid faces with their vertical decomposition and type $b$, $c$ and $d$ points.}
    \label{fig:staticBCD}
\end{figure}

\subparagraph{Type d points.} Recall type $d$ points are intersections
between MinMax and edges bounding $S_k(B\cup R)$. In particular, for
every MinMax edge we care only about its outermost intersection
points. We show how to find the leftmost intersection with a valid
boundary for each segment of MinMax, and can find the rightmost
intersection symmetrically.

We walk along MinMax from left to right through the vertical
decomposition of $S_k(B\cup R)$. We start at the leftmost vertex $v$,
and find what trapezoid contains $v$ in $O(\log n)$ time. The cell has
constant complexity, so we can calculate in $O(1)$ time on which side
the edge of MinMax immediately to the right of $v$ leaves the
cell. There are three options. Either (i) the edge of MinMax ends in
the current trapezoid, and thus produces no type $d$ point, and we
move to the next vertex, (ii) it intersects the edge bounding the cell
on the top or bottom, in which case we have found our leftmost type
$d$ point and we move to the next vertex, or (iii) it moves to an
adjacent cell on the right without intersecting a valid edge, in which
case we continue to walk in the adjacent cell. Since MinMax is
$x$-monotone this procedure enters case (iii) at most once per
trapezoid. There are $O(|S_k(B\cup R)|)$ trapezoids, so we spend
$O(|S_k(B\cup R)|)$ time on the walk. We also perform one point
location query per MinMax vertex, resulting in a total time of
$O(|S_k(B\cup R)| + n\log n)$ time.

\begin{lemma}
  \label{lem:compute_optimal}
Given $S_k(B\cup R)$, we can compute a separator
  $\sopt = \argmin_{s \in S_k(B \cup R)} \Max(s)$ in $O((|S_k(B\cup R)|+n)\log n)$ time.
\end{lemma}

By combining Lemma~\ref{lem:compute_optimal} with the algorithm of
Lemma~\ref{lem:construct_valid_regions_outputsensitive} we thus
immediately obtain an $O((n+|S_k(B\cup R)|+n+k^3)\log^2 n)$ time
algorithm to compute $s_\opt$. Similarly, plugging in
Lemma~\ref{lem:construct_valid_regions_tightk}, we obtain an
$O(k^{4/3} n^{2/3} \log n + (n + k^2) \log n)$ time
algorithm for when $k=k_{\min}$. Plugging in our simple $O(nk+n\log n)$
time algorithm, that constructs the overlay $\A_k(R^*,B^*)$ of
$L_{\leq k}(R^*)$ and $L'_{\leq k}(B^*)$ to find
$S_k(B \cup R)$, we would obtain $O((nk + n)\log n)$ time. We can slightly reduce this to $O(nk+n\log n)$ time as
follows. These results together than thus establish
Theorem~\ref{thm:2d_algorithm}.

\begin{lemma}
  \label{lem:optimal_simple}
  We can compute a separator $s_\opt \in S_k(B \cup R)$ minimizing
  $\Max$ in $O(nk + n\log n)$ time.
\end{lemma}

\begin{proof}
  Rather than binary searching on MinMax to find the edge of MinMax
  below every vertex $v$, we do the following. For each $i \in 1..k$, we
  simultaneously scan through MinMax, and the $i$-level of $R^*$. Both
  MinMax and the $i$-level are $x$-monotone curves. Furthermore, since
  $i \leq k$, we can trace the $i$-level by walking in
  $\A_k(R^*,B^*)$. We do the same for every $n-i$ level of $B^*$. We scan
  through MinMax at most $2k$ times, and we visit every edge of
  $\A_k(R^*,B^*)$ at most once. Thus over all levels, this takes at most
  $O(nk)$ time.

  To compute the type $c$ and $d$ points, we explicitly insert MinMax as well
  as a vertical line through every vertex of MinMax into the overlay
  $\A_k(R^*,B^*)$ of $L_{\leq k}(R^*)$ and $L'_{\leq k}(B^*)$. See
  Figure~\ref{fig:verticalSlabs}. The complexity remains $O(nk)$:
  Every vertical line intersects at most $k$ red, and at most $k$ blue
  lines, and thus adds at most $O(nk)$ vertices. Similarly, within a
  ``slab'' between two consecutive such vertical lines, an edge of
  MinMax can intersect at most $O(k)$ lines. Using Finke and Hinrichs'
  algorithm~\cite{mapoverlay} this takes $O(nk)$ time, and the resulting arrangement clearly contains all type $c$ and $d$ points.
\end{proof}

\section{Maintaining $\Max$ for $k_{\min}$ under restricted insertions}
\label{sec:insertion_only_tight_k}

We want to maintain a valid separator $\sopt$ minimizing $\Max(\sopt)$ under insertions. This turns out to be difficult, so we consider the case when $k = k_{\min}$ (so for a tight $k$), and we impose two restrictions on the insertions:

\begin{enumerate}
    \item The convex hulls of $R$ and $B$ do not change. This has two consequences in the dual: MinMax does not change, and the error of any fixed point does not change.
    \item $k_{\min}$ does not increase after an insertion, i.e. not all valid cells are made invalid by an insertion. This means the valid regions only ever decrease in size.
\end{enumerate}

For each type $a,b,c,d$ (as defined in Lemma \ref{lem:optimum}) we separately maintain an
optimal solution of that type; the overall optimum is then of course
one of these four. For each type we actually maintain a set of
candidate optimal points, and store them in a min-heap with
their error $\Max$ to maintain the point with lowest error.
The type $c$ and especially $d$ vertices are the bottleneck. We obtain the following:

\begin{theorem}
\label{thm:dynamic_tight_k}
Given an initial set of $n$ points $P = B \cup R$, we can build a data structure that maintains an optimal separator $\sopt$ minimizing $\Max(\sopt)$ with $\Mis(\sopt) = k = k_{\min}$ under $m = \Omega((\frac{k^{1/3}}{n^{1/12 + \eps}} + \frac{n^{1/4 - \eps}}{k})\log^{6+\eps})$ restricted insertions in $O(k n^{3/4 + \eps})$ amortized insertion time, using $O((k^{4/3} n^{2/3}+n) \log^5 n)$ space.
\end{theorem}

In \cref{sub:lowerEnvelopeQueries} we first show how to answer lower envelope queries on some set of functions; we will use this result for finding the type $d$ optimum. In the following four sections we consider the insertion of a red line $r$, and show how to maintain an optimal type $a$, $b$, $c$, and $d$ point respectively. Inserting a blue line is handled similarly.

\subsection{Lower-envelope queries in a set of functions}
\label{sub:lowerEnvelopeQueries}
Let
$g(x_0, \dots x_{d_1}, a_0, \dots a_{d_2}) = g(x,a)$ be a
$(d_1 + d_2)$-variate function, for some constants $d_1$ and $d_2$, such that $g(x,a) \leq c$ can be written as a $(d_1 + d_2)$-variate polynomial with constant degree and integer coefficients. Recall that a \emph{semi-algebraic set} is obtained by boolean operations on a collection of such functions, so $g(x,a) \leq c$ itself forms a semi-algebraic set too. In particular it is a constant complexity semi-algebraic set, since it is formed of a constant number (one) of constant-variate, constant-degree functions. Let such a function $g(x,a)$ be an \emph{admissible} function.

We refer
to $x$ as the \emph{variables} and $a$ as the \emph{parameters} of
this function. If we fix a parameter vector $a \in \R^{d_2}$ we get a
$d_1$-variate function $f_a(x) = g(x,a)$. Let $F$ be a set of such
functions, i.e. for each function $f \in F$ there exists some
parameter vector $a_f$ such that $f(x) = g(x, a_f)$. We call
$\R^{d_2}$ the \emph{parameter space}, and observe that each function
$f \in F$ corresponds to a point $a_f \in \R^{d_2}$. Similarly each
point $a \in \R^{d_2}$ corresponds to a function $f_a(x) = g(x,a)$. We
wish to perform lower envelope queries in $F$: given a point
$p \in \R^{d_1}$ we wish to find the lower envelope of $F$ at $p$,
that is, we wish to find the lowest value $\hat{\delta}$ for which
there still exists a function $f \in F$ such that
$f(p) \leq \hat{\delta}$.

\subparagraph{Ranges in parametric space.}
For a given point $p \in \R^{d_1}$ and value $\delta \in \R$, consider the equation $g(p,a) \leq \delta$. All the $a$'s are free variables, and thus this describes a region or \emph{range} $\Gamma_p(\delta)$ in $\R^{d_2}$: a point $a \in \R^{d_2}$ lies in $\Gamma_p(\delta)$ if and only if $f_a(p) \leq \delta$.

\begin{lemma}
  \label{lem:growingRange}
  For any fixed point $p \in \R^{d_1}$ and any two values
  $\delta_1,\delta_2 \in \R$ with $\delta_1 \leq \delta_2$, it holds that
  $\Gamma_p(\delta_1) \subseteq \Gamma_p(\delta_2)$.
\end{lemma}
\begin{proof}
For every point $a \in \Gamma_{\ell}(\delta_1)$ there is a corresponding function $f_a$ s.t. $f_a(p) \leq \delta_1$. Since $\delta_1 < \delta_2$, clearly then $f_a(p) < \delta_2$ also holds, so $a \in \Gamma_p(\delta_2)$, and thus $\Gamma_p(\delta_1) \subseteq \Gamma_p(\delta_2)$. 
\end{proof}

It can be helpful to think of $\delta$ as `time', where $\delta$
starts at $-\infty$ and continuously increases until $\infty$. By
Lemma \ref{lem:growingRange} the range $\Gamma_p(\delta)$ grows in
size as $\delta$ increases; indeed as $\delta$ increases the number of
functions $f \in F$ for which $f(p) \leq \delta$ increases.

For a point $a \in \R^{d_2}$ let $\delta_a$ be such that $a$ lies exactly on the boundary of $\Gamma_p(\delta_a)$. See Figure \ref{fig:specialDeltaValues}. Given $p$ and $a$ we can compute $\delta_a$, and thus the range $\Gamma_p(\delta_a)$, in constant time, by simply solving the equation $f_a(p) = \delta_a$, in which every parameter is known except $\delta_a$.

Let $A = \{a_f \mid f \in F\}$ be the set of parameters (points in
$\R^{d_2}$) corresponding to our functions $F$. Let $\hat{a} = \argmin_{a \in A}\delta_a$, i.e.\ the first point to be contained in the range $\Gamma_p(\delta)$ as $\delta$ increases, and note that none of the other points in $A$ are contained in $\Gamma_p(\delta)$. This point $\hat{a}$ is the point we are looking for, as it corresponds to the function $f_{\hat{a}}$ with the lowest value $f_{\hat{a}}(p) = \delta_{\hat{a}}$ at point $p$.

\begin{figure}
    \centering
    \includegraphics{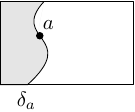}
    \caption{The $\delta$ value $\delta_a$ for a point $a$.}
    \label{fig:specialDeltaValues}
\end{figure}

\subparagraph{The data structure. }
We build a polynomial partition tree $T$ on the points in
$A$~\cite{agarwal2013semialgRange2} (to be precise, we use the
'boundary-fuzzy' variant, which uses symbolic perturbation to ensure
the input set $A$ is not degenerate). Given a point $p \in \R^{d_1}$ and a value $\delta$, we could then construct the range $\Gamma_p(\delta)$ and query the partition tree: if the range contains any point $a \in A$ then there exists a function $f_a$ with $f_a(p) \leq \delta$, if the range is empty such a function does not exist. We could then use parametric search to find the lowest value $\delta$ for which such a function exists, but this would add at least a $\log n$ factor in the query time, which we can avoid as we show below.

The partition tree $T$ recursively partitions the set of points $A$ (see
Figure \ref{fig:partitionTree} for an illustration for when
$d_2=2$). We start at the root node, where the set $A$ is partitioned
into $t = O(r)$ subsets $A_1 \dots A_t$ using an $r$-partitioning
polynomial $h$, for some large but constant value $r$. The polynomial
$h$ partitions the space into \emph{cells} $C_1 \dots C_t$ (shown
in black in the figure). For each cell $C_i$ we get a
\emph{representative point} $p(C_i)$ that lies in $C_i$ (note that
$p(C_i)$ is not generally a point from $A$). We create a child $u_i$
for each subset $A_i$, for which we recursively build a subtree; the
second level of cells is shown in red in Figure \ref{fig:partitionTree}. Once a node
contains only a constant number of points, we create a leaf node, and
the recursion stops.

\begin{figure}
    \centering
    \includegraphics{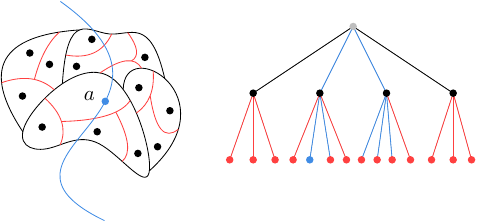}
    \caption{A polynomial partition tree of a set of points, with two levels. }
    \label{fig:partitionTree}
\end{figure}

Consider a cell $C$ and a range $\Gamma_p(\delta)$. We say that $\Gamma_p(\delta)$ \emph{contains} $C$ if $C \subseteq \Gamma_p(\delta)$, and $\Gamma_p(\delta)$ \emph{crosses} $C$ if $C \nsubseteq \Gamma_p(\delta)$ and $C \cap \Gamma_p(\delta) \neq \emptyset$. The most important property of the partition tree is that, for any internal node $u$ of the tree, any range $\Gamma_p(\delta)$ crosses at most $c r^{1 - 1/d_2}$ children of $u$, where $c$ is a constant independent of $r$~\cite{agarwal2013semialgRange2}.

\subparagraph{The query algorithm. } 
Recall that $\hat{a} = \argmin_{a \in A} \delta_a$.
Figure \ref{fig:optimalRange} shows $\Gamma_p(\delta_{\hat{a}})$, and the nodes in the partition tree it crosses.

\begin{figure}
    \centering
    \includegraphics{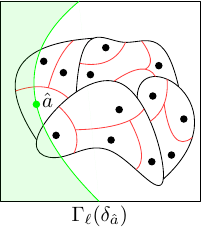}
    \caption{The range $\Gamma_p(\delta_{\hat{a}})$ through $\hat{a}$. }
    \label{fig:optimalRange}
\end{figure}

\begin{lemma}
\label{lem:containCellContainP}
Let $\Gamma_p(\delta)$ be any range. If $\Gamma_p(\delta)$ contains a cell $C$, then it also contains $\hat{a}$.
\end{lemma}
\begin{proof}
It follow from Lemma \ref{lem:growingRange} that the range $\Gamma_p(\delta)$ contains a points $a$ if and only if $\delta_a \leq \delta$. Every cell contains at least one point, by construction of the partition tree, and thus $C$ contains some point $a$ with $\delta_a \leq \delta$. Since $\delta_{\hat{a}} \leq \delta_a$, point $\hat{a}$ is also contained in $\Gamma_p(\delta)$.
\end{proof}

Consider some internal node $u$ with children $u_1 \dots u_t$ (recall that $t = O(r)$), associated with cells $C_1 \dots C_t$. Recall that $p(C_i)$ is a representative point inside cell $C_i$, which is computed and stored during preprocessing for every cell. Let $\Gamma_p(\delta_i)$ be the surface through representative point $p(C_i)$. We sort the children of $u$ by increasing $\delta_i$ value and renumber them, such that $\delta_1 \leq \delta_2 \dots \leq \delta_t$. By Lemma \ref{lem:growingRange}, this also means that $\Gamma_p(\delta_1) \subseteq \Gamma_p(\delta_2) \dots \subseteq \Gamma_p(\delta_t)$.

Let $i'$ be the lowest integer such that $\Gamma_p(\delta_{i'})$ fully contains any cell, and let the contained cell be $C_j$. For a given range $\Gamma_p(\delta)$ we can compute all cells contained in it using the \texttt{A2} operation from~\cite{agarwal2013semialgRange2}. Since we have $O(r)$ cells, and we chose $r$ to be a sufficiently large constant, we can perform this operation for each representative range and thus find $i'$ in $O(1)$ time total. 

From Lemma \ref{lem:containCellContainP} we immediately obtain:

\begin{corollary}
\label{cor:rangeContainsP}
The range $\Gamma_p(\delta_{i'})$ contains $\hat{a}$.
\end{corollary}

We can now prove the last ingredient of our algorithm:

\begin{lemma}
\label{lem:surfaceBoundedCells}
The number of cells of the partition at a node $u$ crossed by or contained in range $\Gamma_p(\delta_{i'}))$ is bounded by $c'r^{1 - 1/d_2}$, where $c'$ is a constant independent of $r$.
\end{lemma}
\begin{proof}
Recall that cell $C_j$ is a cell contained in $\Gamma_p(\delta_{i'})$. Clearly the range $\Gamma_p(\delta_{j})$ through the representative point $p(C_j)$ cannot contain the cell $C_j$ itself, since $p(C_j)$ lies in the interior $C_j$. Therefore, $\Gamma_p(\delta_{j}) \subset \Gamma_p(\delta_{i'})$, and thus $\delta_j < \delta_{i'}$. Since we sorted the $\delta$ values of the children, it follows that $j < i'$, and thus $i' > 1$.

Recall that any range crosses at most $cr^{1 - 1/d_2}$ children of $u$, where $c$ is a constant independent of $r$. Clearly, the number of cells crossed by $\Gamma_p(\delta_{i'})$ is thus bounded by $cr^{1 - 1/d_2}$, and we only need to consider the cells contained in $\Gamma_p(\delta_{i'})$.

Consider the range $\Gamma_p(\delta_{i'-1})$, so the last range to not contain any cell. See Figure \ref{fig:cellsCrossedContainedBounded}. All ranges contained in $\Gamma_p(\delta_{i'})$ must be crossed by $\Gamma_p(\delta_{i' - 1})$. Assume towards a contradiction that this is not the case, so there is some cell $C_k$ that is contained in $\Gamma_p(\delta_{i'})$ but not crossed by $\Gamma_p(\delta_{i'-1})$. Then, by the applying the same reasoning as above to $\Gamma_p(\delta_{i'-1})$ and $\Gamma_p(\delta_{k})$, and to $\Gamma_p(\delta_{k})$ and $\Gamma_p(\delta_{i'})$, we have $\Gamma_p(\delta_{i'-1}) \subset \Gamma_p(\delta_{k}) \subset \Gamma_p(\delta_{i'})$ and thus $\delta_{i'-1} < \delta_k < \delta_{i'}$. This means $k$ must be an integer between $i'-1$ and $i'$, which is not possible.

So, all cells contained in $\Gamma_p(\delta_{i'})$ are crossed by $\Gamma_{\ell}(\delta_{i'-1})$, and thus the number of cells contained in $\Gamma_p(\delta_{i'})$ is also bounded by $cr^{1 - 1/d_2}$. Thus the total number of cells crossed by or contained in $\Gamma_p(\delta_{i'})$ is bounded by $2cr^{1 - 1/d_2} = c'r^{1 - 1/d_2}$, with $c' = 2c$.
\begin{figure}
    \centering
    \includegraphics{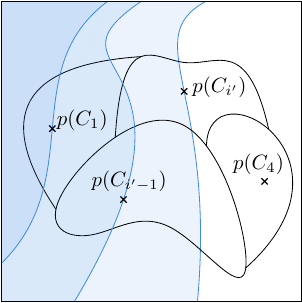}
    \caption{All cells contained in $\Gamma_p(\delta_{i'})$ are crossed by $\Gamma_p(\delta_{i'-1})$. }
    \label{fig:cellsCrossedContainedBounded}
\end{figure}
\end{proof}

The above two lemmas give rise to the following result.

\begin{lemma}
\label{lem:lowerEnvelopeD}
Let $G$ be an admissible $(d_1 + d_2)$-variate function $g(x,a)$ as
described above, let
$A \subset \R^{d_2}$ with $|A| = n$, and let
$F = \{ f_a(x) = g(x,a) \mid a \in A \}$. We can build a data
structure that, given a point $p \in \R^{d_1}$, can find the lower
envelope of the functions $F$ at $p$. This data structure takes $O(n)$
space, can be built in $O(n \log n)$ expected time, and answers
queries in $O(n^{1 - 1/d_2 + \eps})$ time.
\end{lemma}
\begin{proof}
We transform the functions $F$ into a set of points $A$ in $\R^{d_2}$ and build a polynomial partition tree $T$ on $A$, as described above, which uses $O(n)$ space and $O(n \log n)$ expected time. We wish to find point $\hat{a} = \argmin_{a \in A} \delta_a$.

Consider a query with point $p$. We start at the root of $T$. If it is
a leaf we simply iterate through all $O(1)$ points $a$ stored in this
leaf, compute the surfaces $\Gamma_p(\delta_a)$ through them, and
maintain the smallest $\delta_a$ we encounter. If
it is an internal node we find the the first representative range
$\Gamma_p(\delta_{i'})$ to contain a cell as described above. By
Corollary~\ref{cor:rangeContainsP} this range contains
$\hat{a}$. Consider all the cells that are either contained in or
crossed by $\Gamma_p(\delta_{i'})$; clearly $\hat{a}$ lies in one of
these cells. Therefore we recursively search all children that are
contained in or crossed by $\Gamma_{\ell}(\delta_{i'})$. By Lemma
\ref{lem:surfaceBoundedCells} there are at most $c'r^{1 - 1/d_2}$ such
children, each containing at most $O(n/r)$ points.

An internal node contains a constant number of cells (since we choose $r$ to be a large but constant value), and a leaf node contains a constant number of points; therefore we only spend constant time per node we visit. Thus, the total query time $Q(n)$ of the above procedure obeys the following recurrence
\[
Q(n) = \begin{cases}
c'r^{1 - 1/d_2} Q(n/r) + O(1)    & \text{for an internal node} \\
O(n)                        & \text{for a leaf node.}
\end{cases}
\]

This recurrence solves to $Q(n) = O(n^{1 - 1/d_2 + \eps})$, as claimed.
\end{proof}

\subsection{Type a: valid vertices}
\label{sub:dynamicA}
We first consider maintaining the valid cells themselves. We explicitly store and maintain all valid vertices, in $O(k^{4/3} n^{2/3}+n)$ space. For each valid face we store the vertices in clockwise order in a linked list. After the insertion of red line $r$ all valid vertices above $r$ become invalid, and two new valid vertices appear at every valid cell $r$ intersects, see Figure \ref{fig:dynamicAFaces}. By Lemma \ref{lem:lineKIntersections} a line intersects $O(k)$ valid cells, so $O(k)$ new vertices are created. 

\begin{figure}[tb]
    \centering
    \includegraphics{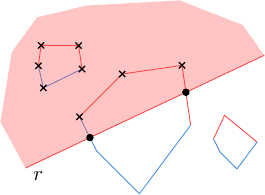}
    \caption{Vertices above $r$ become invalid, vertices below $r$ stay valid. New valid vertices are created at intersected edges.}
    \label{fig:dynamicAFaces}
\end{figure}

We maintain the valid vertices in a dynamic halfplane range reporting data structure~\cite{chanDynamicHalfplaneReporting}. This uses $O((k^{4/3} n^{2/3}+n) \log^5 n)$ space, can be updated in $O(\log^{6 + \eps} n)$ time (for any $\eps > 0$), and can report all $q$ points above a query halfplane in $O(\log n + q)$ time. (With newer techniques, these times can be slightly improved; however, since the type $a$ vertices are not the bottlneck anyway, this will suffice.)

When inserting line $r$ we perform one query to find all $q$ valid vertices $V$ above $r$ in $O(\log n + q)$ time. We first check if $r$ intersects any of the edges adjacent to these vertices $V$; each such intersection corresponds to a new valid vertex. We add these $O(k)$ new vertices to their respective linked lists in $O(k)$ time, and to the range reporting data structure in $O(k \log^{6 + \eps} n)$ time. Afterwards we remove all vertices in $V$ from their linked lists in $O(q)$ time, and from the range reporting data structure in $O(q\log^{6 + \eps} n)$ time. This yields a total insertion time of $O((k + q) \log^{6 + \eps} n)$.

\begin{lemma}
We can maintain the vertices of the valid faces under restricted insertions in $O((k + q) \log^{6 + \eps} n)$ time, where $q$ is the number of valid vertices made invalid by an insertion, using $O((k^{4/3} n^{2/3}+n)\log^5 n)$ space.
\end{lemma}

Due to restriction 1, the error $\Max(p)$ of a fixed point $p$ can not change after an insertion. Thus we can additionally maintain a min-heap containing all valid vertices with their error, to maintain an overall optimal type $a$ point. Since updating a min-heap costs only $O(\log n)$ time, this does not increase the total insertion time.

\begin{lemma}
We can maintain an optimal type $a$ point under restricted insertions in $O((k + q) \log^{6 + \eps} n)$ time, where $q$ is the number of valid MinMax vertices made invalid by an insertion, using $O((k^{4/3} n^{2/3}+n)\log^5 n)$ space.
\end{lemma}

\subsection{Type b: MinMax vertex}

\begin{figure}[tb]
    \centering
    \includegraphics{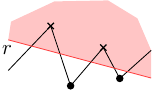}
    \caption{The MinMax vertices above $r$ become invalid. }
    \label{fig:dynamicB}
\end{figure}

By restriction 1 the MinMax curve does not change, nor does the error of its vertices. When inserting line $r$, all valid MinMax vertices above $r$ become invalid. We again build a dynamic halfplane reporting data structure on the valid
vertices, of size $O(n \log^5 n)$. (In fact, this data structure needs to support only deletions, not insertions, which can be handled more simply and quickly. However, since a fully dynamic data structure was already used above, we may as well re-use it.) For every insertion we perform one query in $O(\log n + q)$ time to find all $q$ vertices above $r$, and remove those from the
reporting data structure in $O(q \log^{6+\eps}n)$ time.

As before, we additionally maintain a min-heap containing all valid MinMax vertices with their error to maintain an optimal type $b$ point.

\begin{lemma}
We can maintain an optimal type $b$ point under restricted insertionsin $O(\log n + q \log^{6+\eps}n)$ time, where $q$ is the number of valid vertices made invalid by an insertion, time and $O(n \log^5 n)$ space.
\end{lemma}

\subsection{Type c: the first valid point above/below a MinMax vertex}
We want to maintain the best type $c$ point, i.e.\ the \emph{first} valid point above or below a MinMax vertex. First we define a type $c'$ point to be \emph{any} valid point above or below a MinMax vertex, i.e. without the requirement to be the first.

\begin{lemma}
\label{lem:typec_cprime}
An optimal type $c'$ point is also an optimal type $c$ point.
\end{lemma}
\begin{proof}
Every type $c$ point is also a type $c'$ point, so we only have to prove that an optimal type $c'$ point is also a type $c$ point.

Let $p$ be an optimal type $c'$ point, and assume w.l.o.g. that it lies below a MinMax vertex. See Figure \ref{fig:typecprime}. Assume by contradiction that point $p$ is not a type $c$ point, i.e. $p$ is not the first valid point below a MinMax vertex: then there must be another valid point $q$ between $p$ and MinMax. This point $q$ must also be a type $c'$ point, since it is also vertically below a MinMax vertex. Point $q$ is vertically closer to MinMax than point $p$, so by Lemma \ref{lem:vert} we have $\Max(q) < \Max(p)$, which is a contradiction since we assumed $p$ to be an optimal type $c'$ point.
\end{proof}

\begin{figure}[tb]
    \centering
    \includegraphics{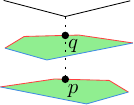}
    \caption{An optimal type $c'$ point is also a type $c$ point. }
    \label{fig:typecprime}
\end{figure}

\subparagraph{Geometry of type \texorpdfstring{$c'$}{c'} points.}

\begin{figure}[tb]
    \centering
    \includegraphics{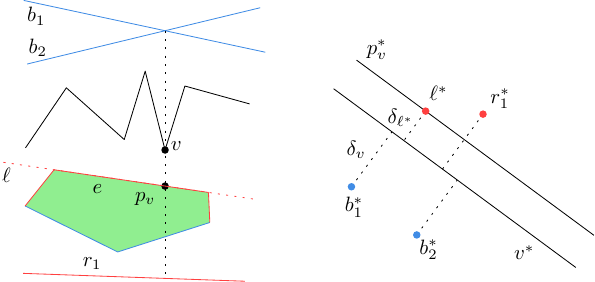}
    \caption{Left: a dual valid face with MinMax in black, a query line $\ell$, and a type $c'$ point $p_v$. Right: the dual of the left, so the primal.}
    \label{fig:dynamicC}
\end{figure}

Consider a MinMax vertex $v$, defined by (w.l.o.g.) two blue lines
$b_1$ and $b_2$ and a red line $r_1$. See Figure
\ref{fig:dynamicC}. Let $e$ be a valid edge below $v$ with supporting line $\ell$. Vertex $v$ projects a type $c'$ point $p_{v,\ell}$ on
$\ell$. We dualize everything, going `back' to the primal plane, as on the right side of the figure. Here, line $v^*$ is defined by points $b_1^*$ and
$b_2^*$ on one side and $r_1^*$ on the other, with equal distance
$\delta_v$ to all three points. Line $\ell$ dualizes to a point $\ell^*$, and the type $c'$ point $p_{v,\ell}$ dualizes to a
line $p^*_{v,\ell}$ through $\ell^*$ parallel to $v^*$. Let $\delta_{\ell} = \dist(\ell^*, v^*)$. Recall that the error $\Max(p^*_{v,\ell})$ of $p^*_{v,\ell}$ is the distance to its furthest misclassified point, which must be one of the points defining $v^*$, so $\Max(p^*_{v,\ell}) = \delta_v + \delta_{\ell^*}$. 

\begin{figure}[tb]
    \centering
    \includegraphics{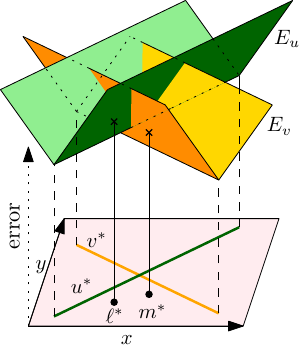}
    \caption{The error function for two lines $v_1^*$ and $v_2^*$, and the corresponding error of two points $\ell_1^*$ and $\ell_2^*$.}
    \label{fig:typeCFunction}
\end{figure}

For a fixed MinMax vertex $v$, each point $\ell^* \in \R^2$ defines a type $c'$ line $p^*_{v,\ell}$ as described above, with some error $\Max(p^*_{v,\ell})$. We define the bivariate error function $E_{v}(\ell^*) = \Max(p^*_{v,\ell})$. Since this error is a constant ($\delta_v$) plus the
distance to the line ($\delta_{\ell^*}$), $E_{v}(\ell^*)$ is piecewise linear. In particular, the graph of $E_{v}(l^*)$ is the union of two halfplanes at a $45^{\circ}$ originating from $v^*$, forming a \emph{wedge} as in Figure \ref{fig:typeCFunction}, starting at height $\delta_v$. For intuition, observe that all points $\ell^*$ on line $p^*_{v,\ell}$ would result in the same type $c'$ line $p^*_{v,\ell}$ and thus the same error, and moving $\ell^*$ a distance $d$ away from $v^*$ would indeed increase its error by $d$.

Fix a valid edge $e$ with supporting line $\ell$, and let $V$ be the set of MinMax vertices that define a type $c'$ point on $e$. With slight abuse of notation, let $E_V = \{ E_v | v \in V \}$ be the set of wedge-functions defined by vertices in $V$. The following lemma follows immediately from the definition of the error function:

\begin{lemma}
\label{lem:optimalPointLowerEnvelope}
Point $p_{\hat{v},\ell}$ is the type $c'$ point on $e$ with the lowest error, if and only if the function $E_{\hat{v}}$ has the lowest value at point $\ell^*$ among all functions $E_V$, i.e. $\hat{v} = \argmin_{v \in V} E_v(\ell^*)$.
\end{lemma}

For example in Figure \ref{fig:typeCFunction}, the function $E_u$ has the lowest value at point $\ell^*$, so vertex $u$ defines the best type $c'$ point on $\ell$.

\subparagraph{Data structures.}

By Lemma \ref{lem:optimalPointLowerEnvelope}, we can find the best type $c'$ point on a given line $\ell$ by finding the error function with the lowest value at point $\ell^*$, i.e.\ by performing a \emph{lower envelope query} in a set of wedge-functions. Simply constructing the entire lower envelope and storing its projection on the $xy$-plane would allow for logarithmic time queries, but would take $O(n^2)$ time and space, and is therefore not feasible. We could use the data structure developed in \cref{sub:lowerEnvelopeQueries} to answer these queries, however we can do something faster in this case, since our functions $E_V$ have a rather simple shape.

Let $E_v$ be a wedge-function defined by line $v^*$. A point $\ell^*$ lies \emph{north of} $E_v$ if $\ell^*$ lies above $v^*$ (in the plane), and similarly $\ell^*$ lies \emph{south of} $E_v$ if $\ell^*$ lies below $v^*$. We first show how to perform restricted lower envelope queries, where all wedges lie on the same side of the query point, and then extend it to the non-restricted case. We use this to find the best type $c'$ point on a query edge, and finally maintain the overall best type $c$ point.

\begin{lemma}
\label{lem:queryTypeCEnvelopeAllAbove}
Let $V$ be a set of $n$ MinMax vertices. We can build a data structure that can answer lower envelope queries on the functions $E_V$ for query points $\ell^*$ that lie south of all wedges. This data structure has size $O(n)$ and answers queries in $O(\log n)$ time.
\end{lemma}
\begin{proof}
Until now we have written $E_{v}(\ell^*) = \Max(p^*_{v,\ell}) = \delta_v + \dist(\ell^*, v^*)$, but we can make it more explicit. For a point $p: (p_x,p_y)$ and a line $m: ax + by + c = 0$ let $\widehat{\dist}(p,m) = |a p_x + b p_y + c| / \sqrt{a^2 + b^2}$ be the \emph{signed} distance between $p$ and $m$, i.e. it is positive if $p$ lies north of $m$ and negative otherwise. Note that, for a fixed $m$, this function is linear in $p$. Then:

$$E_{v}(\ell^*) = 
\begin{cases}
     f_{\text{north}}(\ell^*) = \delta_v + \widehat{\dist}(\ell^*, v^*)  & \text{if $\ell^*$ north of $v^*$} \\
     f_{\text{south}}(\ell^*) = \delta_v - \widehat{\dist}(\ell^*, v^*)  & \text{if $\ell^*$ south of $v^*$}
\end{cases}$$

For a fixed $v^*$, both $f_{\text{north}}$ and $f_{\text{south}}$ are linear in $\ell^*$. Since $\ell^*$ lies south of all wedges, we can simply take $E_v(\ell^*) = f_{\text{south}}(\ell^*)$, and the lower envelope at $\ell^*$ remains the same. Since $E_v(\ell^*)$ is a linear function now, we can build an $O(n)$ space data structure in $O(n \log n)$ time that allows lower envelope queries in $O(\log n)$  time (this corresponds to lower envelope queries on a set of planes in 3D, see e.g.~\cite{bookBerg}). 
\end{proof}

Analogously, the above lemma works for query points that lie north of all wedges.

We can extend the above data structure to handle general query points, that don't necessarily lie south of all wedges, albeit at a much higher query time:

\begin{lemma}
\label{lem:queryTypeCEnvelopeGeneral}
Let $V$ be a set of $n$ MinMax vertices. We can build a data structure that can answer lower envelope queries on the functions $E_V$ for any query point $\ell^*$. This data structure has size $O(n \log n)$, and answers queries in $O(\sqrt{n}\log n)$.
\end{lemma}
\begin{proof}
We build a partition tree~\cite{chanPartitionTrees} of size $O(n)$ on the MinMax vertices $V$. Each node $u$ has an associated canonical subset $V_u \subseteq V$ of points, and on each canonical subset we build the data structure from Lemma \ref{lem:queryTypeCEnvelopeAllAbove}. This uses $O(n \log n)$ space.

For query point $\ell^*$, let $V_n(\ell) \subseteq V^*$ and $V_s(\ell) \subseteq V$ be the set of MinMax vertices that $\ell$ lies north of, respectively that $\ell$ lies south of (by duality, $\ell^*$ thus lies south of all lines $V_n(\ell)^*$). We can query the partition tree to find $O(\sqrt{n})$ nodes representing $V_n(\ell)$, in $O(\sqrt{n})$ time. In each such node, we are certain that $\ell^*$ lies south of all wedges in its canonical subset, so we can query the associate data structure for the lower envelope at $\ell^*$ in $O(\log n)$ time. We similarly find and query the nodes representing $V_s(\ell)$. We maintain and return the lowest point. This takes $O(\sqrt{n} \log n)$ time.
\end{proof}

We can now use the above data structure to find the best type $c'$ point on an edge:

\begin{lemma}
\label{lem:typecEdgeOptimum}
We can build a data structure that, given a query edge $e$, can find the best type $c'$ point on $e$. This data structure has size $O(n \log^2 n)$, and answers queries in $O(\sqrt{n} \log n)$.
\end{lemma}
\begin{proof}
We build a balanced binary tree on the MinMax vertices, sorted by $x$-coordinate, and on each canonical subset of the tree we build the data structure from Lemma \ref{lem:queryTypeCEnvelopeGeneral}. Since the associated data structure on $m$ vertices takes $O(m \log m)$ space, and the total size of all canonical subsets is $O(n \log n)$, this takes $O(n \log^2 n)$ space in total.

For a query edge $e$ with supporting line $\ell$, we query the binary tree for the $O(\log n)$ nodes representing the MinMax vertices in the $x$-interval of edge $e$. These are exactly the vertices that define a type $c'$ point on $e$, and we wish to find the one with the lowest error. We perform lower envelope queries at point $\ell^*$ using the associate data structures on each node, and by Lemma \ref{lem:optimalPointLowerEnvelope} this gives us the type $c'$ point with the lowest error. 

At a node $u$ containing $n$ MinMax vertices we thus spend $O(\sqrt{n}\log n)$ time for querying the associate data structure, and recurse into one child. The query time thus follows the  recurrence $Q(n) = Q(n/2) + O(\sqrt{n}\log n)$, which solves to $Q(n) = O(\sqrt{n} \log n)$.
\end{proof}

We can now finally maintain an overall optimal type $c'$ (and thus type $c$) point:

\begin{lemma}
\label{lem:typeCMaintain}
We can maintain an optimal type $c$ point under restricted insertions in $O(k \sqrt{n} \log n + q \log n)$ time, where $q$ is the number of valid vertices made invalid by an insertion, using $O(n\log^2 n)$ space.
\end{lemma}
\begin{proof}
We will maintain the best type $c'$ point on every valid edge, and maintain a min-heap of all these local optima to maintain the global optimum. By Lemma \ref{lem:typec_cprime} the optimal type $c'$ maintained by the min-heap is also an optimal type $c$ point.

Unsurprisingly, we build the data structure of Lemma \ref{lem:typecEdgeOptimum}, using $O(n \log^2 n)$ space. 

Suppose we insert a red line $\ell$. As we saw when maintaining the valid faces in Section \ref{sub:dynamicA}, some valid edges disappear, some edges get shorter, some edges do not change, and some new edges on $\ell$ appear. We do not have to do anything with the edges that do not change, as the best type $c'$ point on that edge does not change. For each of the new or shortened edges we simply perform one query in the above data structure to find the (new) best type $c'$ point on the edge, and update or insert the entry in the min-heap. By Lemma \ref{lem:lineKIntersections} a line intersects $O(k)$ valid regions, so there are $O(k)$ new or shortened edges. These $O(k)$ queries take $O(k \sqrt{n} \log n)$ time. For the $q$ edges that disappear, we only have to remove their type $c'$ point from the min-heap in $O(\log n)$ time each. The total time for an insertion is $O(q \log n + k \sqrt{n} \log n)$.
\end{proof}

\subsection{Type d: intersection between MinMax and valid cells}
We want to maintain an optimal type $d$ point, i.e. an intersection between a MinMax edge and a valid edge. Our approach is very similar to how we maintained an optimal type $c$ point. We will create a data structure on the MinMax edges that, given a line $\ell$, can find the best type $d$ point on $\ell$. Within this data structure we again want to perform lower envelope queries on a set of error-functions. However these functions are higher dimensional, and their shape is more complicated, so we resort to using general techniques for semi-algebraic range searching. We use this data structure to maintain the best type $d$ point on each valid edge, as before.

\subparagraph{Geometry of type \texorpdfstring{$d$}{d} points. }
\begin{figure}
    \centering
    \includegraphics{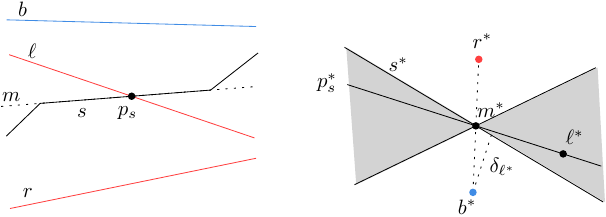}
    \caption{Left: a dual type $d$ point $p_s$, defined by a MinMax edge $e$ and a line $\ell$. Right: the dual of the left, so the primal.}
    \label{fig:dynamicD}
\end{figure}

\begin{figure}
    \centering
    \includegraphics{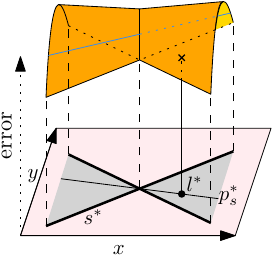}
    \caption{The error function of a wedge $s^*$, and the corresponding error of a point $\ell^*$. The blue line illustrates that this function forms a ruled surface. }
    \label{fig:typeDFunction}
\end{figure}

Consider a MinMax segment $s$ with supporting line $m$, defined by a blue line $b$ and a red line $r$. See Figure \ref{fig:dynamicD}. Let $\ell$ be a line intersecting $s$ in point $p_{s,\ell}$, our type $d$ point. We dualize everything, going 'back' to the primal plane, as on the right side of the figure. Here we have a  double wedge $s^*$ with apex $m^*$, defined by two points $r^*$ and $b^*$ at equal distance from $m^*$. Both $r^*$ and $b^*$ have equal distance to the left and right part of the double wedge $s^*$. Line $\ell$ dualizes to a point $\ell^*$ inside $s^*$, and the type $d$ point $p_{s,\ell}$ dualizes to a line $p_{s,\ell}^*$ through $l^*$ and $m^*$. Some algebra tells us $p_{s,\ell}^*: y = ax + c$, with $a = (m^*_y - \ell^*_y) / (m^*_x - \ell^*_x)$ and $c = \ell^*_y - \ell^*_x a$.

Let $\delta_{\ell} = \dist(b^*, p_{s,\ell}^*) = |a b^*_x - b^*_y + c| / \sqrt{a^2 + 1}$ (or equivalently $\delta_{\ell} = \dist(r^*, p_{s,\ell}^*)$). Recall that the error $\Max(p_{s,\ell}^*)$ of line $p_{s,\ell}^*$ is the distance to its furthest misclassified point, which is $r^*$ or $b^*$, so $\Max(p_{s,\ell}^*) = \delta_{\ell}$.

For a fixed MinMax segment $s$ defined by lines $r$ and $b$, each (primal) point $\ell^*$ defines a type $d$ line $p_{s,\ell}^*$ as described above, with some error $\Max(p_{s,\ell}^*)$. We define the bivariate error function $E_s(\ell^*) = \Max(p_{s,\ell}^*)$, see Figure \ref{fig:typeDFunction} for a schematic illustration. For intuition, observe that any point $\ell^*$ on $p_{s,\ell}^*$ would result in the same type $d$ line $p_{s,\ell}^*$ and thus the same error, and that moving $\ell^*$ towards the outside of wedge $s^*$ will decrease the error.

Let $S$ be the set of MinMax edges intersecting line $\ell$, and let $E_S = \{E_s | s \in S\}$ be the set of functions defined by edges in $S$. The following lemma, very similar to Lemma \ref{lem:optimalPointLowerEnvelope}, holds by definition of the error functions:

\begin{lemma}
\label{lem:typeDLowerEnvelopeBest}
Point $p_{\hat{s},\ell}$ is the type $d$ point on $\ell$ with the lowest error, if and only if function $E_{\hat{s}}$ has the lowest value at point $\ell^*$ among all functions in $E_S$, i.e. $\hat{s} = \argmin_{s \in S} E_s(\ell^*)$.
\end{lemma}

We will thus want to perform lower envelope queries in the set of functions $E_S$.

\subparagraph{Using the lower-envelope data structure with our functions $E_S$.}
We will use the lower-envelope data structure developed in \cref{sub:lowerEnvelopeQueries}. For this, we need to show that our function $E_s$ is admissible. For the function $E_s(\ell^*)$ to be admissible, as defined in \cref{sub:lowerEnvelopeQueries}, we first need to write it in the form $g(x,a)$, i.e. separate the inputs into parameters and variables. Recall that $s$ is defined by two points $r^*$ and $b^*$, and thus $E_s(\ell^*)$ depends on six values: the coordinates of 2D points $r^*$, $b^*$ and $\ell^*$. We wish to find the lower envelope at a given point $\ell^*$. We can thus view $r^*$ and $b^*$ as parameters and $\ell^*$ as the variables, or in other words we can write $E_s = g(x,a)$ where $x = (\ell^*_x, \ell^*_y)$ and $a = (r^*_x, r^*_y, b^*_x, b^*_y )$, so $g(x,a)$ is a $(2 + 4)$-variate function. We now need to show that this function can be written as a constant degree polynomial with integer coefficients.

\begin{lemma}
The equation $E_s(p) = |a p_x - p_y + c| / \sqrt{a^2 + 1} \leq \delta$ can be written polynomial of degree $4$ in $r^*$, $b^*$, and $\ell^*$, with integer coefficients.
\end{lemma}
\begin{proof}

$$E_s(p) = \frac{|a r_x - r_y + c|} {\sqrt{a^2 + 1}} \leq \delta$$

Move square root to other side:

$$|a r_x - r_y + c| \leq \delta \sqrt{a^2 + 1} $$

Square:

$$(a r_x -r_y + c)^2 \leq \delta^2 (a^2 + 1) $$

Working out the square:

$$a^2 r_x^2 + r_y^2 + c^2 - a r_x r_y + a r_x c - r_y c \leq \delta^2 a^2 + \delta^2 $$

Substituting $c = p_y - p_x * a$:

$$a^2 r_x^2 + r_y^2 + (p_y - p_x a)^2 - a r_x  r_y + a r_x (p_y - p_x a) - r_y (p_y - p_x a) \leq \delta^2 a^2 + \delta^2 $$

Working out the square:

$$a^2 r_x^2 + r_y^2 + p_y^2 + p_x^2 a^2 - 2 p_y p_x a - a r_x  r_y + a r_x p_y - p_x r_x a^2 - r_y p_y + r_y p_x a - \delta^2 a^2 - \delta^2 \leq 0$$

Grouping terms:

$$a^2 (r_x^2 + p_x^2 - \delta^2 - p_x r_x) + a(r_x p_y + p_x r_y - r_x r_y - 2 p_y p_x) +  r_y^2 + p_y^2 - r_y p_y  - \delta^2 \leq 0$$

Since $a = (m_y - p_y)/(m_x - p_x)$, multiply by $(m_x - p_x)^2$ to get rid of divisions:

$$(m_y - p_y)^2 (r_x^2 + p_x^2 - \delta^2 - p_x r_x) + (m_y - p_y) (m_x-p_x) (r_x p_y + p_x r_y - r_x  r_y - 2p_y p_x) \dots$$

$$\dots  + (m_x - p_x)^2(r_y^2 + p_y^2 - r_y p_y  - \delta^2) \leq 0$$

Substituting $m_x = \frac{r_x + b_x}{2}$ and $m_y = \frac{r_y + b_y}{2}$:

$$(\frac{r_y + b_y}{2} - p_y)^2 (r_x^2 + p_x^2 - \delta^2 - p_x r_x) + (\frac{r_y + b_y}{2} - p_y) (\frac{r_x + b_x}{2}-p_x) (r_x p_y + p_x r_y - r_x r_y - 2p_y p_x) \dots$$

$$\dots + (\frac{r_x + b_x}{2} - p_x)^2(r_y^2 + p_y^2 - r_y p_y  - \delta^2) \leq 0$$

We can view $p$ and $\delta$ as constants, so the polynomial is of degree $4$ in $r$, $b$, and $\ell$.
\end{proof}

By Lemma~\ref{lem:lowerEnvelopeD} we can thus perform lower envelope queries in $E_S$ in $O(n^{3/4 + \eps})$ time.

\begin{lemma}
\label{lem:typeDDatastructure}
We can build a data structure that, given a query line $\ell$, can find the best type $d$ point on $\ell$. This data structure has size $O(n \log^2 n)$, and answers queries in $O(n^{3/4 + \eps})$ time.
\end{lemma}
\begin{proof}
This is a three-level data structure, where the first two levels are `regular' non-polynomial partition trees~\cite{chanPartitionTrees}, and the last level is our data structure from Lemma~\ref{lem:lowerEnvelopeD}.

More specifically, the level 1 partition tree is built on all MinMax vertices, and can return a set of canonical subsets representing all endpoints of MinMax segments that lie above $\ell$. On each canonical subset we build a level 2 partition tree, which can return a set of canonical subsets representing intersected segments whose other endpoint lies below $\ell$. Lastly, on these canonical subsets we build our level 3 data structure from Lemma~\ref{lem:lowerEnvelopeD} to find the lower envelope of the functions they define. This gives us the lower envelope of all functions defined by segments intersecting $\ell$, which by Lemma~\ref{lem:typeDLowerEnvelopeBest} gives us the best type $d$ point on $\ell$.

The canonical subsets of a partition tree have total size $O(n \log n)$, and since we have two levels of explicit canonical subsets this results in $O(n \log^2 n)$ space. 

Let $Q(n)$ be the query time of our entire multi-level data structure, and let $Q_2(n)$ and $Q_3(n)$ be the query times of the level 2 and level 3 data structure. By Lemma \ref{lem:lowerEnvelopeD}, $Q_3(n) = O(n^{3/4+\eps})$. Consider the level 2 partition tree, and consider a node $u$ containing $n$ points. Query line $\ell$ crosses $c\sqrt{r})$ children of $u$ on which we recurse (for some constant $c$), and we directly query the level 3 data structure on all $O(r)$ cells fully below $\ell$. Therefore the level 2 query time $Q_2(n)$ follows the following recurrence:

$$
Q_2(n) = \begin{cases}
c\sqrt{r}Q_2(n/r) + r Q_3(n)      & \text{for an internal node} \\
Q_3(n) = O(1)                       & \text{for a leaf node}
\end{cases}
$$

Using the Master Theorem we can see that $r Q_3(n)$ dominates the subproblems, and thus $Q_2(n) = r Q_3(n) = O(n^{3/4 + \eps})$. The overall query time $Q(n)$ follows an identical recurrence, and thus $Q(n) = O(n^{3/4 + \eps})$. 
\end{proof}

We need one more level: we have a data structure that can answer queries for a line $\ell$, but we need a data structure that can answer queries for a valid edge $e$.

\begin{lemma}
\label{lem:typeDDatastructureEdge}
We can build a data structure that, given a query edge $e$, can find the best type $d$ point on $e$. This data structure has size $O(n \log^3 n)$, and answers queries in $O(n^{3/4 + \eps})$ time.
\end{lemma}
\begin{proof}
This data structure is analogous to the one used for type $c$ points in Lemma \ref{lem:typecEdgeOptimum}. We build a balanced binary search tree on the MinMax vertices sorted by $x$-coordinate, and build the data structure from Lemma \ref{lem:typeDDatastructure} on each canonical subset, using $O(n \log^3 n)$ space.

We handle a query identically: for a query edge $e$ with supporting line $\ell$ we find nodes representing the MinMax vertices in the $x$-interval of edge $e$, and query their associate data structure for the best type $d$ point in that $x$-interval. The query time follows the recurrence $Q(n) = Q(n/2) + O(n^{3/4 + \eps})$ which solves to $Q(n) = O(n^{3/4 + \eps})$.
\end{proof}

\begin{lemma}
We can maintain an optimal type $d$ point under restricted insertions in $O(k n^{3/4 + \eps} + q \log n)$ time, where $q$ is the number of valid vertices made invalid by an insertion, using $O(n \log^2 n)$ space. 
\end{lemma}
\begin{proof}
We can do this the exactly how we maintained an optimal type $c$ point in Lemma \ref{lem:typeCMaintain}, by maintaining the best type $d$ point on every valid edge, and using a min-heap to maintain the global optimum. After an insertion, for each of the $O(k)$ new or shortened edges we perform one query on the data structure from Lemma \ref{lem:typeDDatastructureEdge}, in $O(k n^{3/4 + \eps})$ total time. For the $q$ edges that disappear, we only have to remove their type $d$ point from the min-heap in $O(q \log n)$ total time.
\end{proof}

\subsection{Total amortized insertion time}
We can separately maintain the best point of each type $a,b,c,d$ using the above data structures; if an insertion makes $q_1$ valid vertices invalid and makes $q_2$ valid MinMax vertices invalid, then the insertion time is $O((k + q_1)\log^{6+\eps} n + \log n + q_2 \log^{6 + \eps} n + k \sqrt{n} \log n + q_1 \log n + kn^{3/4 + \eps} + q_1 \log n) = O(kn^{3/4+\eps} + q_1 \log^{6 + \eps} + q_2 \log^{6+\eps})$. Over a sequence of $m$ insertions, let $Q_1$ be the sum of all $q_1$'s, so the total number of valid vertices made invalid. Define $Q_2$ similarly. The total time of all these insertions is $O(mkn^{3/4+\eps} + Q_1 \log^{6 + \eps} + Q_2 \log^{6+\eps})$. 

For any insertion $q_1$ can be as large as $O(k^{4/3}n^{2/3} + n)$, since every valid vertex can be made invalid by a single insertion. However, since a vertex can become invalid only once and only $O(k)$ new vertices appear every insertion, $Q_1 = O(k^{4/3}n^{2/3} + n + mk)$. Similarly, although any $q_2$ can be $O(n)$, their sum $Q_2 = O(n)$ as well. Since $Q_2$ is smaller then $Q_1$, we focus only on $Q_1$.

We wish to find a value $m$ such that $O(mkn^{3/4+\eps})$ dominates $O(Q_1\log^{6+\eps})$; then the amortized insertion time would simply be $O(kn^{3/4+\eps})$. Some simple math tells us this holds for $m = \Omega(\frac{(k^{4/3}n^{2/3} + n) \log^{6+\eps}}{kn^{3/4 + \eps}}) = \Omega((\frac{k^{1/3}}{n^{1/12 + \eps}} + \frac{n^{1/4 - \eps}}{k})\log^{6+\eps})$. This finalizes the proof of \cref{thm:dynamic_tight_k}.

\section{\texorpdfstring{$\eps$}{e}-Approximation}
\label{sec:eps-Approximation}

Let $\sopt \in S_k(B\cup R)$ be an optimal valid separator minimizing
$\Max$, and let $\eps \in (0,1)$ be some given threshold. Our goal is
to compute a \emph{$(1 + \eps)$-approximation} of $\sopt$: that is, we
want to find a valid separator $\hat{s}$ with
$\Max(\hat{s}) \leq (1+\eps)\Max(\sopt)$. The main idea is to replace the Euclidean distance function $\dist$ by some convex distance function $\hat{d}$ that approximates $\dist$, and compute a separator $\hat{s}$ that minimizes $\hat{M}(s) = \max_{p \in X(s,B \cup R)} \hat{d}(p,s)$.

\subsection{Convex distance function}
Let $p$ be a point and $s$ be a line, let $t = \Theta(1/\sqrt{\eps})$, and let $T$ be a convex regular $t$-gon
centered at the origin inscribed by a unit disk. See Figure
\ref{fig:unitDiskTgon}. We then define the convex distance function
$\hat{d}(p,s) = \min\{\lambda \mid s \cap (p+\lambda T) \neq \emptyset
\}$ to be the smallest scaling factor for which a scaled copy of $T$
centered at $p$ intersects $s$. Observe that this distance is realized
in a corner $v$ of the $t$-gon; i.e. the $t$-gon scaled by a factor
$\hat{d}(p,s)$ intersects $s$ in a corner point $v$ of the $t$-gon,
see Figure \ref{fig:realizer}. We say $v$ is a \emph{realizer} for the
line~$s$. As $t$ increases, $T$ becomes more circular, and $\hat{d}$ and $\dist$ become more similar. It can be shown that $\dist(p,s) \leq \hat{d}(p,s) \leq (1 + \eps) \dist(p,s)$~\cite{dudley1974metricEntropyEpsilon,har2019proofOfDudley}, and thus $\Max(s) \leq \hat{M}(s) \leq (1+\eps)\Max(s)$. It follows that the separator $\hat{s}$ minimizing $\hat{M}$ is a $(1+\eps)$-approximation of $\sopt$.

\begin{figure}[tb]
\begin{minipage}{0.5\textwidth}
    \centering
    \includegraphics{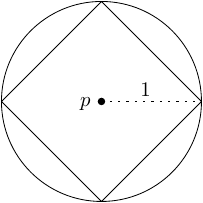}
    \caption{The unit circle and unit $4$-gon}
    \label{fig:unitDiskTgon}
\end{minipage}%
\begin{minipage}{0.5\textwidth}
    \centering
    \includegraphics{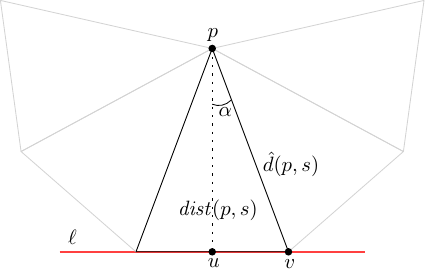}
    \caption{A cone of a $t$-gon around point $p$ aligned with line $s$}
    \label{fig:epsilonProof}
\end{minipage} 
\end{figure}

\begin{figure}[tb]
    \centering
    \includegraphics{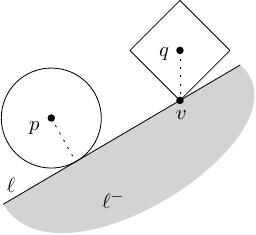}
    \caption{The Euclidean distance from point $p$ to halfspace
      $\ell^-$, and the convex (Manhattan) distance from point $q$ to halfspace
      $\ell^-$ with $t = 4$.
    }
    \label{fig:realizer}
\end{figure}

\subsection{Overview}

Each corner $v$ of the $t$-gon corresponds to some interval $J_v$ of slopes for which it is a realizer. These slope intervals $J_v$, over all corners of the $t$-gon, partition all possible slopes into $t$ intervals. For each slope interval $J_v$ we will compute an optimal valid separator $\hat{s}^v$ with slope in $J_v$, and finally $\hat{s} = \min_v \hat{s}^v$. From now on, we consider one such slope interval $J_v$. Assume w.l.o.g. that $v$ is vertically below the center point of the $t$-gon (we can rotate the plane to achieve this). This means interval $J_v$ is centered at slope $0$, so $J_v = (-\pi / t, \pi /t)$.

The distance between a point and a line is now simply their vertical distance. More formally, let $r = (r_x,r_y) \in R$ be a red point, and let $s : y = mx + c$ be a line with $m \in J_v$. By filling in $r$ in the equation of $s$ we find the signed (vertical) distance between $r$ and $s$ is $d_r(s) = r_y - (m r_x + c)$. Recall red points are misclassified if they lie in the halfplane $s^+$ above $s$. Thus the error of point $r$ is its distance to $s^-$, which is $\hat{M}_r(s) = \max\{0,d_r(s)\}$. Similarly, the error of a blue point $b \in B$ is $\hat{M}_b(s) = \max\{0,d_b(s)\}$ where $d_b(s) = m b_x + c - b_y$. Hence, our goal is to compute a separator $\hat{s}^v$ with $m \in J_v$ that minimizes total error $\hat{M}(s) = \max_{p \in B \cup R} d_p(s)$, while misclassifying at most $k$ points. 

We will be working in the dual. We are only interested in separators with slope in $J_v$, which restricts us to a vertical slab from $x = -\pi / t$ to $x = \pi /t$. Recall that the dual transformation preserves vertical distance between points and lines, so for a dual point $s$ the error $\hat{M}(s)$ is the vertical distance from $s$ to the furthest misclassified line.

We first create a data structure that can quickly
find a valid separator (a dual point) $s$ with $\hat{M}(s) \leq \delta$, for a given value
$\delta$, if it exists. We then use parametric search to find the optimal value $\delta$, and a separator $\hat{s}^v$ with that error.

\subsection{Decision problem for a given \texorpdfstring{$\delta$}{d}}
For a given value $\delta$, we want to quickly find a valid
separator $s$, with slope in $J_v$ and
$\hat{M}(s) \leq \delta$, if it exists. We extend Chan's algorithm from Section \ref{sec:chansAlgorithm} to this end.

Recall that the error of a point is its vertical distance to
the furthest misclassified line, which lies
either on $L'_0(B)$ or $L_0(R)$. Therefore, all points with error at most $\delta$ lie at most $\delta$ below $L'_0(B)$. This can be imagined as moving $L'_0(B)$ down by $\delta$. Let the resulting chain be the \emph{convex $\delta$-chain}. See Figure \ref{fig:deltaChainsDeltaRegion}. Similarly, let the \emph{concave $\delta$-chain} be $L_0(R)$ moved up by $\delta$. All points with error at most $\delta$ must thus lie above the convex $\delta$-chain, and below the concave $\delta$-chain: let the \emph{$\delta$-region} be the intersection of these two regions, a convex polygon of complexity $O(n)$. Since we are interested only in separators with slope in $J_v$, we clip the $\delta$-region to the vertical slab induced by $J_v$, after which it is still convex.

We do not need to consider any points $s$ outside the $\delta$-region, as they either have a slope not in $J_v$ or have error $\hat{M}(s^*) > \delta$. So the question becomes: does there exist a valid point in the $\delta$-region? 

\begin{lemma}
\label{lem:approxConvexConcaveOptima}
If any valid point exists in the $\delta$-region, there must exist a valid point on the intersection of a convex chain (a blue chain or the convex $\delta$-chain) and a concave chain (a red chain or the concave $\delta$-chain).
\end{lemma}
\begin{proof}
Let $s$ be the valid point, and let $C$ be the valid region containing $s$. Let $C'$ be the intersection of $C$ and the $\delta$-region. Both $C$ and the $\delta$-region are bounded on top by concave chains, and bounded on the bottom by convex chains. Therefore, $C'$ is also bounded on top by concave chains, and bounded on the bottom by convex chains. In particular, the leftmost point of $C'$ lies on a convex-concave intersection, and by construction it is valid and lies in (the boundary of) the $\delta$-region, proving the lemma.
\end{proof}
The set of concave-convex intersections consist of $O(k^2)$ red-blue intersections, $O(k)$ intersections of the concave $\delta$-chain with blue chains, $O(k)$ intersections between the convex $\delta$-chain with red chains, and $O(1)$ intersections between the $\delta$-chains. Candidate convex-concave intersections in Figure \ref{fig:deltaRegionCandidatePoints} are marked.

We can calculate all $O(k^2)$ red-blue intersections during preprocessing, among them find the valid point with smallest error $p_{min}$, and simply forget about all others. The $O(k)$ intersections involving the $\delta$-chains can not be preprocessed, as they require a query value $\delta$. Explicitly constructing the $\delta$-chains for a query would take $O(n)$ time, as the envelopes have complexity $O(n)$, so we do it implicitly by just adding $\delta$ to the $y$-coordinate of $L_0(R)$ (or subtracting $\delta$ from the $y$-coordinate of $L'_0(B)$) whenever we need this value in a calculation.

\begin{figure}[tb]
\begin{minipage}{.45\textwidth}
    \centering
    \includegraphics{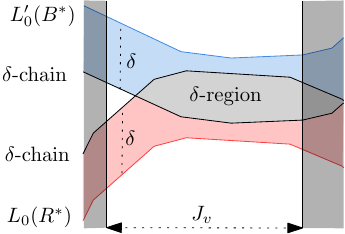}
    \caption{An overview of the geometry: the slab induced by $J_v$, the $\delta$-chains created by moving $L'_0(B)$ down and $L_0(R)$ up by $\delta$, and the $\delta$-region formed by them}
    \label{fig:deltaChainsDeltaRegion}
\end{minipage}
\hspace{0.1\textwidth}
\begin{minipage}{.45\textwidth}
    \centering
    \includegraphics{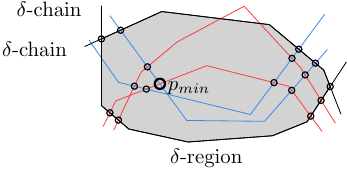}
    \caption{A $\delta$-region, with some red and blue chains. Candidate intersections are marked.}
    \label{fig:deltaRegionCandidatePoints}
\end{minipage}
\end{figure}

\subparagraph{Data structure.}
The data structure consists of three parts:
\begin{itemize}
    \item A concave chain decomposition of $L_{\leq k}(R)$, and a convex chain decomposition of $L'_{\leq k}(B)$, with a chromatic ply data structure for every chain.
    \item The point $p_{min}$, the red-blue intersection with lowest error.
    \item The envelopes $L_0(R)$ and $L'_0(B)$.
\end{itemize}
This can be constructed in $O((n + k^2) \log n)$ time using Chan's method as explained in Section \ref{sec:chansAlgorithm}, and uses $O(n + k^2)$ space.

\subparagraph{Query.}
We answer a query with value $\delta$ as follows:
\begin{itemize}
\item Check if $\hat{M}(p_{\min}) \leq \delta$. If so, return $p_{\min}$. This takes $O(1)$ time.
\item Find the $O(k)$ convex-concave intersections involving the $\delta$-chains. This takes $O(k \log n)$ time. 
\item Build a red ply data structure for the convex $\delta$-chain, and build a blue ply data structure for the concave $\delta$-chain. This takes $O(k \log k)$ time.
\item For each intersection, calculate whether it is valid or not using the chromatic ply data structures on the chains, and calculate its error. This takes $O(k \log k)$ time.
\item Return the valid intersection with lowest error if its error is at most $\delta$, otherwise there exists no point with error at most $\delta$. This takes $O(k)$ time.
\end{itemize}

\begin{lemma}
\label{lem:decisionDelta}
We can build a data structure that, given a query value $\delta \in \R$, can
find a valid separator $s$ with slope in $J_v$ such that $\hat{M}(s^*) \leq \delta$, if it
exists. This data structure takes $O(n + k^2)$ space
and $O((n + k^2)\log n)$ time to build, and queries are answered in $O(k \log n)$ time.
\end{lemma}

In fact we can parallelize these queries with $k$ processors; this will be useful later for parametric search.

\subparagraph{Parallel query with $k$ processors}
We answer a query with value $\delta$ as follows:
\begin{itemize}
\item Check if $M(p_{\min}) \leq \delta$. If so, return $p_{\min}$. This step does not need to be parallelised. $O(1)$ time.
\item Find the $O(k)$ convex-concave intersections involving the $\delta$-chains. Since all intersections are independent, with $k$ processors this takes $O(\log n)$ time.
\item Build a red ply data structure for the convex $\delta$-chain, and build a blue ply data structure for the concave $\delta$-chain. Since the ply data structure simply consists of two sorted lists, this comes down to sorting $k$ values, which takes $O(\log k)$ time with
  $k$ processors~\cite{cole88paral_merge_sort}. 
\item For each intersection, calculate whether it is valid or not using the chromatic ply data structures on the chains, and calculate its error. Again these calculations are independent, so this takes $O(\log k)$ time.
\item Return the valid point with lowest error if its error is at most $\delta$, otherwise there exists no point with error at most $\delta$. With a simple divide and conquer approach this takes $O(\log k)$ time.
\end{itemize}

\begin{lemma}
\label{lem:decisionDeltaPara}
Using the data structure of Lemma \ref{lem:decisionDelta}, with $k$ processors, we can answer queries in $O(\log n)$ time.
\end{lemma}

\subsection{Finding optimal \texorpdfstring{$\delta$}{d} (parametric search)}
In Lemma \ref{lem:decisionDelta} we give a decision algorithm: for a given value $\delta$, does there exist a valid point with error at most $\delta$? Using \emph{parametric search}~\cite{megiddo1983parametric, agarwal1998parametricSearch} we can turn this decision algorithm into an optimisation algorithm: what is the lowest value $\delta^*$ for which there exists a valid point? Parametric search requires that the decision algorithm (1) behaves discontinuously at the optimal value $\delta^*$, and (2) only computes roots of polynomials in $\delta$ of small degree to govern its control flow. Clearly condition (1) holds for our problem, since for any $\delta < \delta^*$ there does not exist a valid point with error at most $\delta$, while for any $\delta \geq \delta^*$ there does. To see that condition (2) holds, observe that our algorithm finds intersections between convex and concave chains, sorts lists of points, and binary searches on lists or chains. All these steps yield low degree polynomials. 

Parametric search requires a parallel decision algorithm that uses $p$ processors and runs in $T_p$ parallel steps, as well as a sequential algorithm that runs in $T_s$ time, and computes $\delta^*$ in $O(p T_p + T_p T_s \log p)$. Lemma \ref{lem:decisionDelta} gives a sequential algorithm running in $T_s = O(k \log n)$ time (after building the data structure in $O((n + k^ 2)\log n)$ time) and Lemma \ref{lem:decisionDeltaPara} gives a parallel algorithm running in $T_p = O(\log n)$ time with $p = k$ processors. This yields $O(k \log n + \log n k \log n \log k) = O(k \log^2 n \log k)$ time for the parametric search. Including the time to build the data structure this yields an $O((n + k^2)\log n + k \log^2 n \log k) = O((n + k^2)\log n)$ time algorithm for finding $\delta^*$.

\begin{lemma}
\label{lem:optimalDelta}
We can find a valid separator $\hat{s}^v$ with slope in $J_v$ that minimizes $\hat{M}(\hat{s}^v)$ in $O((n + k^2) \log n)$ time.
\end{lemma}

Recall that this is for one slope interval $J_v$. Doing this for all $t = \Theta(1/\sqrt{\eps})$ intervals results in the following:

\approximationAlgorithm*

\subsection{Maintaining an approximate separator}
\label{sub:Maintaining_an_approximate_separator}
Recall the data structure from Lemma \ref{lem:decisionDelta} which, for a given slope interval $J_v$ and value $\delta$, can find if there exists a valid separator in $J_v$ with error at most $\delta$. Also recall Lemma \ref{lem:semiOnlineChains}, which tells us we can maintain a concave chain decomposition of the $\leq k$ level of a set of lines. By combining the two, we can maintain an approximate optimal solution to the k-Mis MinMax problem.

We can maintain the data structure from Lemma \ref{lem:decisionDelta} in the semi-online setting in $O(k \log^3 n)$ update time: we maintain the chains, chromatic ply data structures, and point $p_{\min}$ using Lemma \ref{lem:semiOnlineIntersections}, and additionally maintain the envelopes of $R$ and $B$. After every insertion we perform the parametric search from Lemma \ref{thm:2d_approximation_algorithm} to find a new approximate optimum. The only difference is that we now have $O(k \log n)$ chains, instead of $O(k)$, which changes the runtime slightly. The $\delta$-chains now intersect $O(k \log n)$ other chains, resulting in a sequential time $T_s = O(k \log^2 n)$, and a parallel time that remains $T_pO(\log n)$ but now using $p = O(k \log n)$ processors. Filling in the runtime formula for parametric search now gives $O(p T_p + T_p T_s \log p) = O(k \log n \log n + \log n k \log^2 n \log \log k) = O(k \log^3 n \log \log k)$ time.

Doing this for all $t = \Theta(1/\sqrt{\eps})$ slope intervals $J_v$ results in the following:

\dynamicApproximation*

\bibliography{references}

\clearpage
\appendix

\end{document}